\documentclass[journal]{IEEEtran}
\usepackage{amsmath,amssymb,mathtools,bbm}
\usepackage{graphicx}
\usepackage[caption=false,font=footnotesize]{subfig}
\usepackage{cite}
\usepackage[T1]{fontenc}
\usepackage[utf8]{inputenc}
\newtheorem{proposition}{Proposition}

\renewcommand{\P}{\mathsf{P}}
\newcommand{\indicator}[1]{\mathbbm{1}\{#1\}}
\newcommand{\bfa}{\mathbf{a}}
\newcommand{\bfv}{\mathbf{v}}
\newcommand{\bfh}{\mathbf{h}}
\newcommand{\bfU}{\mathbf{U}}
\newcommand{\bfPhi}{\boldsymbol{\Phi}}
\newcommand{\bfQ}{\mathbf{Q}}
\newcommand{\bfJ}{\mathbf{J}}
\newcommand{\bfH}{\mathbf{H}}
\newcommand{\bfR}{\mathbf{R}}
\newcommand{\bfP}{\mathbf{P}}
\newcommand{\bfK}{\mathbf{K}}
\newcommand{\bfI}{\mathbf{I}}
\newcommand{\bfx}{\mathbf{x}}
\newcommand{\bfn}{\mathbf{n}}
\newcommand{\bfz}{\mathbf{z}}
\newcommand{\bfSig}{\boldsymbol{\Sigma}}
\newcommand{\bfmu}{\boldsymbol{\mu}}
\newcommand{\bfX}{\mathbf{X}}
\newcommand{\bfW}{\mathbf{W}}
\newcommand{\bfy}{\mathbf{y}}
\newcommand{\bfw}{\mathbf{w}}
\newcommand{\bfr}{\mathbf{r}}
\newcommand{\bfxi}{\boldsymbol{\xi}}
\newcommand{\bfPi}{\boldsymbol{\Pi}}
\newcommand{\bfeps}{\boldsymbol{\varepsilon}}
\newcommand{\dtv}{d_{\mathrm{TV}}}

\begin{document}

\title{Generative Diffusion Models for Secure UAV-ISAC with Real-Time Target Tracking}

\author{Homa~Nikbakht,
        Yonina~C.~Eldar,~\IEEEmembership{Fellow,~IEEE},
        and~H.~Vincent~Poor,~\IEEEmembership{Life Fellow,~IEEE}%
}

\maketitle

\begin{abstract}
Unmanned aerial vehicle (UAV) platforms with integrated sensing and communication (ISAC) are prime candidates for surveillance missions, in which a dual-function drone tracks a mobile target and alerts a ground station when the target enters a protected region. Such alerts generate event-triggered communication traffic. By observing the change in the traffic, an eavesdropper (Eve) can therefore learn in real time whether the target is detected inside the protected region without decoding any message. Existing physical-layer security techniques protect the message content or hide that a transmission occurs, but the composition of the traffic remains exposed. We formulate the threat as traffic-composition covertness and propose a keyless diffusion-based transmission scheme that hides the triggered traffic. The scheme superposes the beamformed information streams of the UAV with a message-agnostic safeguard waveform generated by a conditional denoising diffusion implicit model. Training shapes the safeguard distribution such that triggered and non-triggered slots become statistically indistinguishable at Eve. The legitimate user shares no secret key with the UAV.
We formulate the joint design as a rate maximization under covertness, reliability, and sensing constraints, with the sensing accuracy measured by the posterior Cram\'er-Rao bound (PCRB) of the beam-dependent tracking problem.
Our numerical results show that, with multiple transmit antennas and an Eve signal-to-noise ratio (SNR) of $5$~dB, a learned Eve stays near the blind guess of $0.5$ for every safeguard beam split, a detector that knows the safeguard model reaches $0.71$ to $0.73$ inside Bob's main lobe and $0.53$ to $0.55$ outside it, while the average position PCRB varies only between $0.59$ and $0.70$~m$^2$.

\end{abstract}

\begin{IEEEkeywords}
Integrated sensing and communication, UAV communications, physical-layer security, covert communication, generative AI, diffusion models, posterior Cram\'er-Rao bound.
\end{IEEEkeywords}

\section{Introduction}\label{sec:intro}

Integrated sensing and communication (ISAC) enables a wireless system to sense and communicate with shared hardware, antennas, and spectrum, and is a key candidate technology for sixth-generation (6G) networks \cite{Liu2022, Visa2024, Ahmadi2026, NikbakhtISIT24}. Unmanned aerial vehicles (UAVs) offer mobility, rapid deployment, and line-of-sight-dominated air-to-ground channels. A single dual-function transceiver on a UAV can serve ground users while tracking mobile targets in real time \cite{Qin2025, Luat2025, NikbakhtISIT25, Ahmed2025}. UAV-ISAC is therefore suitable for surveillance and protection missions, in which a drone monitors a restricted area and reports its sensing outcomes to a ground station. In such missions the two functions are coupled in both directions. The echo returns update the communication beam, as in radar-assisted predictive beamforming \cite{Liu2020PB} and learning-based ISAC adaptation \cite{NikbakhtAsilomar24}, and the sensing outcomes determine the communication traffic, because the system sends alarms and tracks updates when it detects a target.

Physical-layer security for ISAC has been studied in the literature \cite{Su2021, Boroujeni2026, Goel2008, Bash2013, Zhou2019, Tang2024,  Yao2025}. Secrecy-oriented designs protect the message content or shape the transmit beampattern against malicious sensing targets \cite{Su2021}, and randomized frequency-hopping waveforms hide the radar parameters from a passive eavesdropper (Eve) \cite{Boroujeni2026}. Artificial-noise beamforming masks the information signal at Eve with a jamming component \cite{Goel2008}. Covert communication hides the existence of a transmission from a warden \cite{Bash2013, Zhou2019}, and recent extensions embed the hidden communication inside a dual-functional ISAC waveform \cite{Tang2024}. For UAV platforms, secure ISAC designs jointly optimize beamforming and trajectory against eavesdroppers \cite{Yao2025}.

These techniques leave the traffic pattern unprotected. In surveillance missions, the communication traffic is event-triggered, and a sensing outcome can change the transmission in several ways. The transmitter may send an additional message, increase the data rate, or reallocate power once the target enters a protected region. Any change of the traffic that depends on the sensing outcome is itself sensing information. By observing the change, Eve learns whether and when the system has detected the target without decoding any transmission. Content protection does not address this leakage, because Eve ignores the content, and existence-hiding does not either, because the transmitter is always on. The quantity to protect is the \emph{composition} of the traffic across time slots. To the best of our knowledge, no published work addresses this threat.

Hiding the traffic composition has three requirements. First, Eve's observations in triggered and non-triggered slots must follow the same distribution. The masking therefore needs to be random, with a trainable distribution that makes the two slot types indistinguishable. Second, the scheme should operate without a shared secret key, since distributing and protecting a key on a UAV that can be captured are operational risks. The legitimate user (Bob) must then decode through masking randomness that cannot be canceled. Third, the masking must be generated using the available transmit resources, and covertness, sensing accuracy, and communication reliability therefore compete within a single transmit signal.

In this paper, we propose a secure transmission scheme for a multiple-input single-output (MISO) UAV-ISAC system that addresses all three requirements. The core of the scheme is a denoising diffusion implicit model (DDIM) \cite{Song2021}, a generative artificial intelligence (AI) model that transforms Gaussian noise into samples of a trainable distribution through a sequence of denoising steps \cite{Ho2020}. Conditioned only on public context, the DDIM draws new initial noise in every slot and outputs an unpredictable, message-agnostic safeguard waveform, and hence serves as a randomized encoder. Training shapes the safeguard distribution such that Eve's observation becomes approximately invariant to whether a slot carries an actual aperiodic message or dummy bits. The safeguard is thus an artificial noise in the sense of \cite{Goel2008} whose entire distribution is a design variable, and provides the trainable masking that handles the first requirement. Existing uses of generative AI in physical-layer security are surveyed in \cite{ZhaoGenAI2024}. In the surveyed works, the generator produces transmission strategies such as beamforming or artificial-noise solutions rather than the transmitted signal itself. Closest to our approach, in \cite{Wang2025} diffusion-generated signals are modulated onto the transmitted pilots to prevent unauthorized sensing of user activity from the channel state. In our design, by contrast, the generator output itself is the transmitted randomness and hides the communication traffic rather than the channel.

Under this scheme, Bob  also shares no seed with the UAV which addresses the second requirement.
The third requirement shapes the transmit design because the masking, sensing illumination, and message transmission share a single transmit signal while satisfying a power constraint.
 This setting thus creates a communication-sensing-covertness trade-off that no existing work has studied.
The main contributions of this work are as follows.
\begin{itemize}
\item We develop a UAV-ISAC system equipped with a uniform linear array (ULA) which tracks a mobile target from its echo returns and serves Bob over the same downlink. The downlink carries a periodic stream in every slot and an additional aperiodic message whenever the UAV declares that the target is inside the protected region. The event-triggered message on top of an always-on stream leads to a new threat model, termed traffic-composition covertness.
\item Our main contribution is a keyless diffusion-based safeguard that hides the traffic the trigger generates. We superpose the beamformed information streams with a message-agnostic DDIM safeguard whose trained distribution makes triggered and non-triggered slots statistically indistinguishable at Eve, with the safeguard distribution as the design variable. We measure covertness in total variation and train it by minimizing the maximum mean discrepancy \cite{Gretton2012} together with the loss of an adversarial classifier \cite{Goodfellow2014}, while the public aperiodic codebook is optimized jointly with the safeguard.
\item We formulate the joint communication-sensing-security design as a rate maximization under covertness, reliability, and sensing constraints. The sensing accuracy is measured by the posterior Cram\'er-Rao bound (PCRB) of the target tracking problem. The PCRB also yields the trigger rule for the aperiodic traffic, computed up to a single one-dimensional integral over the angle. The problem is solved by penalty-based end-to-end training. Since Bob shares no key with the UAV, we design a denoising decoder that Bob trains on public safeguard samples.
\end{itemize}
Our numerical analysis validates the design. The trigger reproduces the target's true presence in the protected region in $120$ of $123$ in-region slots with no false activations.
Shifting the safeguard power between masking Eve and illuminating the region changes the average position PCRB only between $0.59$ and $0.70$~m$^2$, while the learned Eve stays near the blind-guess level of $0.5$ and the model-aware Eve between $0.71$ and $0.73$ when located inside Bob's main lobe. In the single-antenna case, where Eve and Bob observe the same signal, the safeguard holds the learned Eve at or below $0.55$ for Eve signal-to-noise ratios (SNRs) from $0$ to $12$~dB, reaches $0.69$ to $0.80$ without the safeguard.

The paper is organized as follows. Section~\ref{sec:model} introduces the system model, the threat model, and the problem formulation. Section~\ref{sec:sensing} presents the tracking, PCRB, and trigger design. Section~\ref{sec:scheme} describes the keyless diffusion-based transmission and the receivers. Section~\ref{sec:objective} presents the training procedure that solves the design problem. Section~\ref{sec:numerical} represents the numerical analysis and Section~\ref{sec:conclusion} concludes the paper.

\section{System Model and Problem Formulation}\label{sec:model}

\begin{figure}[t]
\centering
\includegraphics[width=0.7\columnwidth]{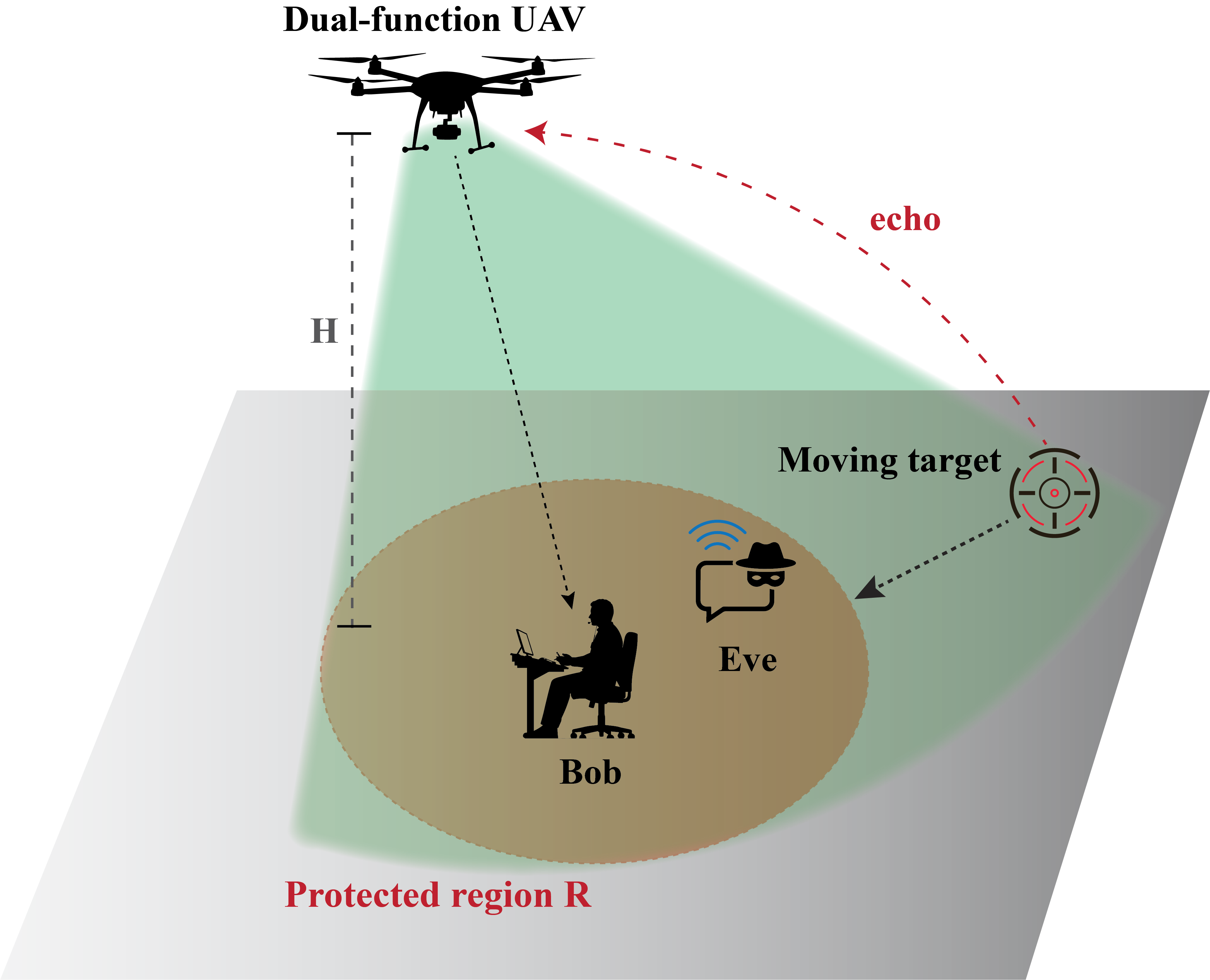}
\caption{Illustration of the system model. A dual-function UAV tracks a moving target while communicating with Bob. The downlink carries a periodic stream in every TS, and an aperiodic message is added only when the target lies inside the protected region. Eve observes the downlink and attempts to detect the triggered transmissions.}%
\label{fig:system}
\end{figure}

\subsection{Geometry and Traffic Model}
Consider a UAV-assisted ISAC system, illustrated in Fig.~\ref{fig:system}, in which a dual-function UAV communicates with a legitimate user (Bob) in the presence of an eavesdropper (Eve), and simultaneously as a mono-static radar tracks a mobile ground target\footnote{We consider a single passive single-antenna Eve and a single point target; multi-antenna wardens and multi-target scenarios are left for future work.}. The UAV hovers at altitude $H$ above the origin of a three-dimensional Cartesian system\footnote{The UAV position is fixed throughout; joint trajectory and energy-aware designs \cite{Ahmed2025} are left for future work.} and carries a half-wavelength ULA with $N$ elements aligned with the $x$-axis, used for both transmission and echo reception. Bob, Eve, and the target are on the ground at horizontal positions $\mathbf b = [x_b, y_b]$, $\mathbf e = [x_e, y_e]$, and $\mathbf p_t[\ell] = [x_t[\ell], y_t[\ell]]$, respectively. Bob and Eve are static while the target moves. The total flight time $T$ is divided into $L$ time slots (TS) of duration $\delta := T/L$, and $n$ symbols are transmitted per TS.

The UAV communicates two traffic types to Bob: \emph{periodic} traffic, transmitted over the entire flight, and \emph{aperiodic} traffic, triggered only when the target is predicted to lie within a protected region $\mathcal R$, modeled as a disc of radius $R_0$ centered at Bob. Eve's goal is to determine which TSs carry aperiodic traffic and thereby infer whether the target has entered $\mathcal R$.

For a ground position $\mathbf p = [x_p, y_p]$, the direction cosine of the line of sight from the UAV with respect to the array axis is
\begin{equation}\label{eq:dcos}
u(\mathbf p) := \frac{x_p}{\sqrt{H^2 + \|\mathbf p\|^2}},
\end{equation}
and the corresponding steering vector is $\bfa(u) = [1, e^{j\pi u}, \ldots, e^{j\pi (N-1) u}]^T$, where $N$ is the number of ULA elements at the UAV. We write $u_b := u(\mathbf b)$, $u_e := u(\mathbf e)$, and $u_t[\ell] := u(\mathbf p_t[\ell])$. The UAV knows Bob's direction accurately and has access only to a nominal value $\hat u_e$ of Eve's direction cosine. The true value is modeled as
\begin{equation}\label{eq:evecsimodel}
u_e = \hat u_e + \delta_e, \qquad |\delta_e| \le \Delta_e,
\end{equation}
where $\delta_e$ is an unknown error with known maximum magnitude $\Delta_e$. The case $\Delta_e = 0$ corresponds to perfect knowledge of Eve's direction.

\subsection{Dual-Function Transmit Signal}\label{sec:txmodel}
The transmit signal of each TS composes of three unit-average-power waveforms in $\mathbb C^{1\times n}$ and three spatial directions. We introduce them in turn and state their combination in \eqref{eq:U} at the end of this subsection.

\subsubsection{Periodic Waveform $\bfX_{p,\ell}$}
The periodic stream carries a new message in every TS. The message $m_{p,\ell} \in \{1, \ldots, M_p\}$ is encoded as $\bfX_{p,\ell} = f_p(m_{p,\ell})$, where $f_p$ is the encoder of a public short-blocklength channel code whose $M_p$ codewords are equal-energy sequences on the power shell $\{\bfx : \|\bfx\|^2 = n\}$. The periodic stream is always present and  is independent of the sensing outcome. It represents the regular service traffic of the downlink, defines the stationary emission profile inside which the event-triggered traffic must hide, and also provides the radar illumination of the protected region.

\subsubsection{Aperiodic Waveform $\bfX_{a,\ell}$}
The aperiodic stream carries the event-triggered traffic and is the security-critical component of the transmit signal. Its content is governed by the binary trigger state $a[\ell] \in \{0,1\}$, which the sensing layer sets to $a[\ell] = 1$ when it declares the target inside the protected region $\mathcal R$ (Section~\ref{sec:sensing}). Since an on-off transmission reveals the trigger immediately, the UAV transmits the stream in every TS and selects the index as
\begin{equation}\label{eq:mtilde}
\tilde m_{a,\ell} = \begin{cases} m_{a,\ell}, & a[\ell] = 1,\\ d_\ell, & a[\ell] = 0,\end{cases}
\end{equation}
where $m_{a,\ell}$ denotes an actual aperiodic message and the dummy index $d_\ell$ is drawn uniformly from a public dummy set, a standard first step in covert communication \cite{Bash2013}. The map $\bfX_a(\cdot)$ assigns to each index a deterministic codeword from a public codebook, and $\bfX_{a,\ell} := \bfX_a(\tilde m_{a,\ell})$.
The codewords are \emph{power-equalized}, with $\|\bfX_a(m)\|^2 = n$ for every index $m$ including the dummy ones, and no index can be identified from the codeword energy. The codewords are, however, not placed at random on the power shell as random placement forms message and dummy clusters that Eve separates reliably. The placement is therefore a design variable that we optimize jointly with the safeguard in Section~\ref{sec:objective}, and the optimized codewords remain on the shell to preserve power equalization.

\subsubsection{Safeguard Waveform $\bfW_{\theta,\ell}$}
Continuous transmission hides the on-off pattern and power equalization hides the codeword energy, yet the trigger still leaks. Each codeword is a fixed point in $\mathbb C^{1\times n}$ from a public codebook. Eve can correlate her received signal with all codewords, identify the most likely one, and check whether it belongs to the message set or to the dummy set, a test that succeeds with high probability at moderate SNR. The trigger state thus leaks through the \emph{positions} of the codewords in signal space, which we refer to as the codeword geometry.

The safeguard removes this leak by adding randomness that Eve cannot subtract. At every TS, the UAV draws new Gaussian noise, passes the noise through a conditional denoising diffusion implicit model (DDIM) \cite{Ho2020, Song2021} with parameter $\theta$, conditioned on a public tag $\psi_\ell$ of slot parameters, and superimposes the resulting waveform $\bfW_{\theta,\ell}$ on the transmission. The waveform is \emph{message-agnostic}, drawn independently of $\tilde m_{a,\ell}$, and its realization therefore carries no trigger information. The safeguard distribution is shaped during an offline training phase (Section~\ref{sec:objective}) and is public once training is complete. Covertness does not require hiding the distribution. Training shapes the distribution such that message and dummy codewords induce nearly the same observation distribution at Eve, and the indistinguishability persists even though Eve knows the trained distribution exactly, because the realization drawn in each TS remains unknown. The safeguard is thus an artificial noise \cite{Goel2008} whose distribution is learned. The scheme is \emph{keyless}. No secret seed is shared between the UAV and Bob, and Bob knows the safeguard only through the public pair $(\theta, \psi_\ell)$.

\subsubsection{Spatial Directions}
Each of the three waveforms is transmitted along a unit-norm direction in $\mathbb C^{N}$, and a ground receiver at direction cosine $u$ observes the waveform with an amplitude proportional to the inner product between $\bfa(u)$ and the transmit direction. A transmit direction that adapts to the sensing output is itself a side channel. If a beam follows the tracked target, the gain at Eve varies with the target state, and the received power alone reveals the track and hence the trigger. We therefore require every transmit direction to be \emph{measurement-independent}, i.e., statistically independent of the sensing measurements, and we sense by illuminating the protected region as a whole rather than by pointing a beam at the target which is in line with low-probability-of-intercept radar principles \cite{Stove2004, Pace2009}.

Bob must decode the aperiodic waveform $\bfX_{a,\ell}$ within its own TS, and the waveform is therefore transmitted along the maximum-ratio direction $\bfa(u_b)/\sqrt N$, which maximizes the received power at Bob. The periodic waveform $\bfX_{p,\ell}$ is transmitted along the \emph{region-illumination beam} $\bfw_{\mathcal R}$, a fixed unit-norm vector whose beampattern is approximately constant over the direction-cosine set $\mathcal U_{\mathcal R} := \{u(\mathbf p) : \mathbf p \in \mathcal R\}$ of the protected region. With $u_1, \ldots, u_K$ denoting a uniform grid of $\mathcal U_{\mathcal R}$, we obtain $\bfw_{\mathcal R}$ by beampattern matching \cite{Stoica2007},
\begin{equation}\label{eq:wr}
(\bfw_{\mathcal R}, \mu^\star) := \arg\min_{\|\bfw\| = 1,\; \mu > 0}\; \sum_{k=1}^{K} \Big( \big|\bfa(u_k)^{\mathsf H} \bfw\big| - \mu \Big)^2,
\end{equation}
which selects the unit-norm vector whose response magnitude stays closest to a common level $\mu$ at every grid direction, i.e., a flat sector beam. We solve \eqref{eq:wr} by alternating least squares. The beam is computed once, offline, from the public region geometry, and illuminates every in-region position at a near-uniform gain.

The safeguard direction must satisfy three requirements. It must stay away from Bob, to leave the decoding undisturbed. It should reach Eve, where the safeguard hides the codeword geometry. It can also reinforce the region illumination, because the echo SNR grows with the total power arriving at the target. We therefore construct the safeguard direction from the region-illumination beam $\bfw_{\mathcal R}$ and from the steering vector $\bfa(\hat u_e)$ of Eve's \emph{nominal} direction. Projecting both onto the orthogonal complement of $\bfa(u_b)$ and normalizing to unit norm yields the projected region direction $\tilde\bfv_r$ and the projected Eve direction $\tilde\bfv_e$, given by
\begin{equation}\label{eq:vtilde}
\tilde \bfv_r := \frac{\bfPi_b\, \bfw_{\mathcal R}}{\big\| \bfPi_b\, \bfw_{\mathcal R} \big\|}, \qquad \tilde \bfv_e := \frac{\bfPi_b\, \bfa(\hat u_e)}{\big\| \bfPi_b\, \bfa(\hat u_e) \big\|},
\end{equation}
where $\bfPi_b := \bfI - \bfa(u_b)\bfa(u_b)^{\mathsf H}/N$ denotes the projector. Among all directions with a spatial null at Bob, $\tilde\bfv_r$ and $\tilde\bfv_e$ are the closest to $\bfw_{\mathcal R}$ and to $\bfa(\hat u_e)$, respectively. Combining the two with a weight $\omega \in [0,1]$ and normalizing gives the safeguard direction
\begin{equation}\label{eq:beams}
\bfv(\omega) := \frac{\sqrt{\omega}\, \tilde\bfv_r + \sqrt{1-\omega}\, \tilde\bfv_e}{\big\| \sqrt{\omega}\, \tilde\bfv_r + \sqrt{1-\omega}\, \tilde\bfv_e \big\|}.
\end{equation}
Every such combination satisfies $\bfa(u_b)^{\mathsf H} \bfv(\omega) = 0$. This spatial null follows the artificial-noise principle of \cite{Goel2008} and makes Bob's received signal safeguard-free in \eqref{eq:ybeff}. With $\omega = 1$ all safeguard energy reinforces the region illumination, and with $\omega = 0$ the full energy points at Eve for maximum masking.
All three directions are fixed over the flight, and none depends on the sensing output.

\subsubsection{Composite Transmit Signal}
The UAV superimposes the three waveforms, with the share $\lambda \in [0,1]$ dividing the power budget between the periodic stream and the aperiodic slot and $\gamma \in (0,1]$ dividing the aperiodic power between the codeword and the safeguard. The transmitted $N \times n$ signal of TS~$\ell$ is
\begin{IEEEeqnarray}{rCl}\label{eq:U}
\bfU_\ell &=& \sqrt{\lambda \P}\, \bfw_{\mathcal R}\, \bfX_{p,\ell} + \sqrt{(1-\lambda)\gamma \P}\, \frac{\bfa(u_b)}{\sqrt N}\, \bfX_{a,\ell} \IEEEnonumber\\
&& +\; \sqrt{(1-\lambda)(1-\gamma)\P}\; \bfv(\omega)\, \bfW_{\theta,\ell},
\end{IEEEeqnarray}
where $\P$ is the power budget. All waveforms have unit average power and all directions unit norm, and the average power constraint $\frac{1}{n}\mathbb E[\|\bfU_\ell\|_F^2] \le \P$ holds with equality.

\subsection{Communication Channels}\label{sec:commchannels}
Bob's and Eve's channels are line-of-sight, $\bfh_b = g_b \bfa(u_b)$ and $\bfh_e = g_e \bfa(u_e)$, with complex path gains $g_b$ and $g_e$. The UAV knows Bob's channel exactly. For Eve, the UAV knows the path gain only up to an SNR interval, published through the tag $\psi_\ell$, and the direction only through the nominal value $\hat u_e$ in \eqref{eq:evecsimodel}. At the end of TS~$\ell$, the received signals at Bob and Eve are
\begin{IEEEeqnarray}{rCl}
\bfy_{b,\ell} &=& \bfh_b^{\mathsf H} \bfU_\ell + \bfz_{b,\ell}, \qquad \bfz_{b,\ell} \sim \mathcal N_{\mathbb C}(0, \sigma_b^2 \bfI),\label{eq:yb}\\
\bfy_{e,\ell} &=& \bfh_e^{\mathsf H} \bfU_\ell + \bfz_{e,\ell}, \qquad \bfz_{e,\ell} \sim \mathcal N_{\mathbb C}(0, \sigma_e^2 \bfI).\label{eq:ye}
\end{IEEEeqnarray}
Substituting \eqref{eq:U} and using $\bfa(u_b)^{\mathsf H} \bfv(\omega) = 0$ gives
\begin{IEEEeqnarray}{rCl}
\bfy_{b,\ell} &=& g_b^* \sqrt{N} \Big( \beta_p \sqrt{\lambda \P} \bfX_{p,\ell} + \sqrt{(1-\lambda)\gamma\P} \bfX_{a,\ell} \Big) + \bfz_{b,\ell}, \IEEEeqnarraynumspace \label{eq:ybeff}\\
\bfy_{e,\ell} &=& g_e^* \sqrt{N} \Big( \alpha_p \sqrt{\lambda \P} \bfX_{p,\ell} + \alpha_c \sqrt{(1-\lambda)\gamma\P} \bfX_{a,\ell} \IEEEnonumber\\
&& \quad +\; \alpha_s(\omega) \sqrt{(1-\lambda)(1-\gamma)\P}\, \bfW_{\theta,\ell} \Big) + \bfz_{e,\ell}, \label{eq:yeeff}
\end{IEEEeqnarray}
where $(\cdot)^*$ denotes complex conjugation, $\beta_p := \bfa(u_b)^{\mathsf H} \bfw_{\mathcal R}/\sqrt N$ and $\alpha_p := \bfa(u_e)^{\mathsf H} \bfw_{\mathcal R}/\sqrt N$ are the gains of the region beam at Bob and Eve, and $\alpha_c := \bfa(u_e)^{\mathsf H} \bfa(u_b) / N$ and $\alpha_s(\omega) := \bfa(u_e)^{\mathsf H} \bfv(\omega)/\sqrt N$ are the leakage coefficients of the information and safeguard beams at Eve, all with magnitude at most one.

Bob's observation \eqref{eq:ybeff} is thus \emph{safeguard-free}, and the covertness randomization adds no interference at Bob. Eve's observation \eqref{eq:yeeff} contains the safeguard scaled by $\alpha_s(\omega)$, which couples covertness to the beam split. We consider the regime in which Eve lies \emph{inside} Bob's main lobe, where the information stream leaks strongly toward Eve. Since the codeword must be beamformed toward Bob, no choice of transmit directions can weaken this leak, and only the safeguard can provide covertness.

\subsection{Radar Echo and Measurement Model}\label{sec:echomodel}
The target echo received at the UAV array during TS $\ell$ is
\begin{equation}\label{eq:echo}
\mathbf Y_{s,\ell}(t) = \beta\, \bfa(u_t[\ell]) \bfa(u_t[\ell])^{\mathsf H} \bfU_\ell(t - \tau[\ell])\, e^{j 2\pi f_D t} + \mathbf Z_{s,\ell}(t),
\end{equation}
where $t$ denotes continuous time within the slot, $\beta$ represents the reflection coefficient and the two-way path loss, $\tau[\ell] = 2 d_{ut}[\ell]/c$ is the round-trip delay at the UAV-target distance $d_{ut}[\ell] = \sqrt{H^2 + \|\mathbf p_t[\ell]\|^2}$, $f_D$ is the Doppler shift, and $\mathbf Z_{s,\ell}$ is additive white Gaussian noise.

The transmit beampattern determines how much power reaches the target. The average power that \eqref{eq:U} radiates toward a direction $u$ has its largest value $N\P$ when the entire budget is beamformed toward $u$. Normalizing this radiated power by its maximum $N\P$ yields the normalized beampattern
\begin{IEEEeqnarray}{rCl}\label{eq:beampattern}
B(u) &:=& \frac{\mathbb E\big[\|\bfa(u)^{\mathsf H} \bfU_\ell\|^2\big]}{n\, N \P} \IEEEnonumber\\
&=& \lambda\, \frac{|\bfa(u)^{\mathsf H} \bfw_{\mathcal R}|^2}{N} + (1-\lambda)\gamma\, \frac{|\bfa(u)^{\mathsf H} \bfa(u_b)|^2}{N^2} \IEEEnonumber\\
&& +\; (1-\lambda)(1-\gamma) \frac{|\bfa(u)^{\mathsf H} \bfv(\omega)|^2}{N} \;\in\; [0,1].
\end{IEEEeqnarray}
The three terms are the periodic power arriving at $u$ through the region-illumination beam $\bfw_{\mathcal R}$, the aperiodic codeword power, and the safeguard power. All three beams are fixed by the measurement-independent design, the beampattern does not change over the flight, and the sensing quality varies only because the target traverses different directions of this fixed beampattern. After receive beamforming toward the target and matched filtering of \eqref{eq:echo}, the echo SNR is given by \cite[eq.~(4)]{Stoica2007}, \cite[eq.~(2.11)]{Richards2014}
\begin{equation}\label{eq:snrecho}
\mathrm{SNR}_{s}[\ell] = \varsigma\, B(u_t[\ell]) \left( \frac{d_{\mathrm{ref}}}{d_{ut}[\ell]}\right)^{4},
\end{equation}
where the constant $\varsigma$ is the echo SNR at the reference range $d_{\mathrm{ref}}$ under full illumination.

In every TS, the UAV estimates the target range $r$, radial velocity $\dot r$, and azimuth $\varphi$ from the matched-filter output. The resulting sensing observation is
\begin{equation}\label{eq:meas}
\bfy_s[\ell] = h(\bfx[\ell]) + \bfn_s[\ell], \qquad \bfn_s[\ell] \sim \mathcal N(\mathbf 0, \bfR_s[\ell]),
\end{equation}
where $h(\bfx) = [\, r(\bfx),\ \dot r(\bfx),\ \varphi(\bfx) \,]^T$ maps the target state $\bfx$ to the measurement vector and $\bfn_s[\ell]$ is the estimation error. The error variances of the delay, Doppler, and angle estimators scale inversely with the echo SNR in \eqref{eq:snrecho} \cite{Kay1993, Richards2014}. When the echo SNR equals the reference value $\varsigma$, the range, radial-velocity, and azimuth errors have variances $\sigma_{r,0}^2$, $\sigma_{\dot r,0}^2$, and $\sigma_{\varphi,0}^2$. The covariance of the estimation error $\bfn_s[\ell]$ at any echo SNR then follows from the inverse scaling and is given by
\begin{equation}\label{eq:R}
\bfR_s[\ell] = \frac{\varsigma}{\mathrm{SNR}_{s}[\ell]}\, \mathrm{diag}\big(\sigma_{r,0}^2,\ \sigma_{\dot r,0}^2,\ \sigma_{\varphi,0}^2\big).
\end{equation}

\subsection{Threat Model and Covertness Metric}\label{sec:covertness}
Eve knows the protocol, the public codebooks including the dummy index set, the denoiser weights $\theta$, the tag sequence $\psi_\ell$, and the channel $\bfh_e$. Eve does not know the per-slot random draws, namely the transmitted index, the safeguard realization, and the noise. From the observation $\bfy_{e,\ell}$, Eve performs the per-slot binary hypothesis test
\begin{equation}\label{eq:evetest}
\mathcal H_0:\ a[\ell] = 0 \qquad \text{vs.} \qquad \mathcal H_1:\ a[\ell] = 1.
\end{equation}
Let $P^{(0)}_\psi$ and $P^{(1)}_\psi$ denote the distributions of $\bfy_{e,\ell}$ under the two hypotheses given the tag $\psi$; under \eqref{eq:evecsimodel} they also depend on the direction error $\delta_e$. A detector maps $\bfy_{e,\ell}$ to a guess $\hat a \in \{0, 1\}$ and is characterized by its false-alarm probability $\Pr[\hat a = 1 | \mathcal H_0]$ and missed-detection probability $\Pr[\hat a = 0 | \mathcal H_1]$.

We measure the distinguishability of the two hypotheses by the total variation distance, defined for distributions $P$ and $Q$ on a common space as
\begin{equation}\label{eq:tvdef}
\dtv(P, Q) := \sup_{\mathcal A} \big| P(\mathcal A) - Q(\mathcal A) \big|,
\end{equation}
where the supremum runs over all events $\mathcal A$. A detector that decides $\mathcal H_1$ on the set $\mathcal A := \{\mathbf y : \hat a(\mathbf y) = 1\}$ has false-alarm probability $P^{(0)}_\psi(\mathcal A)$ and missed-detection probability $1 - P^{(1)}_\psi(\mathcal A)$, and \eqref{eq:tvdef} yields
\begin{IEEEeqnarray}{rCl}\label{eq:tvbound}
\Pr[\hat a = 1 | \mathcal H_0] + \Pr[\hat a = 0 | \mathcal H_1] &=& 1 - \big( P^{(1)}_\psi(\mathcal A) - P^{(0)}_\psi(\mathcal A) \big) \IEEEnonumber\\
&\ge& 1 - \dtv\big(P^{(0)}_\psi, P^{(1)}_\psi\big).
\end{IEEEeqnarray}
For no detector the sum of the two error probabilities can fall below $1 - \dtv\big(P^{(0)}_\psi, P^{(1)}_\psi\big)$ \cite{Poor}.%
 We therefore impose the covertness constraint
\begin{equation}\label{eq:covconstraint}
\dtv\big(P^{(0)}_\psi, P^{(1)}_\psi\big) \le \epsilon_E, \qquad \forall \psi,\ |\delta_e| \le \Delta_e,
\end{equation}
for a covertness level $\epsilon_E \in [0, 1)$, required to hold for every tag and every admissible direction error. When the two hypotheses are equiprobable, \eqref{eq:tvbound} bounds Eve's probability of error as
\begin{equation}\label{eq:everror}
\Pr[\hat a \neq a] = \tfrac{1}{2} \Pr[\hat a = 1 | \mathcal H_0] + \tfrac{1}{2} \Pr[\hat a = 0 | \mathcal H_1] \ge \frac{1 - \epsilon_E}{2}.
\end{equation}
As $\epsilon_E \to 0$, no detector at Eve improves on a blind guess.

\subsection{Problem Formulation}\label{sec:problem}
 Let $R_a := \frac{1}{n}\log_2 M_a$ denote the per-slot aperiodic rate with $M_a$ non-dummy indices, and let $\epsilon_p, \epsilon_a$ denote the decoding error probabilities of the periodic and aperiodic streams at Bob, with reliability targets $\bar\epsilon_p$ and $\bar\epsilon_a$. The joint design of the generator $\theta$, the aperiodic codebook $\bfX_a(\cdot)$, the power shares $(\lambda, \gamma)$, and the safeguard beam split $\omega$ is
\begin{subequations}\label{eq:opt}
\begin{IEEEeqnarray}{rCl}
\max_{\theta,\, \bfX_a(\cdot),\, \lambda,\, \gamma,\, \omega} && \quad R_a \\
\text{s.t.} && \quad \dtv\big(P^{(0)}_{\psi}, P^{(1)}_{\psi}\big) \le \epsilon_E, \quad \forall \psi,\ |\delta_e| \le \Delta_e, \label{eq:optcov} \IEEEeqnarraynumspace\\
&& \quad \epsilon_p \le \bar\epsilon_p, \qquad \epsilon_a \le \bar\epsilon_a, \label{eq:optrel}\\
&& \quad \frac{1}{L} \sum_{\ell=1}^{L} \mathrm{PCRB}_{\mathrm{pos}}[\ell] \le \epsilon_s, \label{eq:optsense}\\
&& \quad \tfrac{1}{n}\mathbb E\big[\|\bfU_\ell\|_F^2\big] \le \P, \qquad \forall \ell. \label{eq:optpow}
\end{IEEEeqnarray}
\end{subequations}
Here $\mathrm{PCRB}_{\mathrm{pos}}[\ell]$ denotes the position PCRB of the tracking problem, a lower bound on the mean-square position error of any estimator, derived in Section~\ref{sec:pcrb}, and $\epsilon_s$ is the sensing budget, the largest tolerated position PCRB averaged over the $L$ TSs. The power constraint \eqref{eq:optpow} holds by construction. The sensing constraint \eqref{eq:optsense} depends on $(\lambda, \gamma, \omega)$ through \eqref{eq:snrecho}--\eqref{eq:R} and the PCRB recursion of Section~\ref{sec:pcrb}; the covertness constraint \eqref{eq:optcov} depends on the same variables through Eve's observation \eqref{eq:yeeff} and on $\theta$ through the safeguard distribution.  Increasing the beam split $\omega$ moves safeguard energy from Eve to the region illumination; the echo SNR \eqref{eq:snrecho} rises and the sensing constraint relaxes, while the masking term $\alpha_s(\omega)$ in \eqref{eq:yeeff} shrinks and the covertness constraint tightens. Decreasing $\omega$ has the reverse effect, and no choice of $\omega$ relaxes both constraints at once.

The solution operates in two phases. Offline, training returns the safeguard generator $\theta$, the aperiodic codebook, and Bob's decoder for the operating point $(M_a, \lambda, \gamma, \omega)$ (Section~\ref{sec:objective}); online, at each TS, the UAV updates the tracker and the trigger (Section~\ref{sec:sensing}), selects the index by \eqref{eq:mtilde}, draws the safeguard (Section~\ref{sec:ddim}), and transmits \eqref{eq:U}, while Bob decodes with the receivers of Section~\ref{sec:receivers}.

\section{Target Tracking and PCRB-Based Triggering}\label{sec:sensing}
At every TS, the sensing layer must decide whether the target is inside the protected region $\mathcal R$. Section~\ref{sec:ekf} models the target motion and presents the extended Kalman filter (EKF) that turns the noisy observations $\bfy_s[\ell]$ into an estimate of the target position and velocity. Section~\ref{sec:pcrb} derives the PCRB, a lower bound on the estimation error of any estimator, which serves as an estimator-independent measure of confidence in the estimate. Section~\ref{sec:trigger} combines the predicted position with the PCRB-based confidence measure into the trigger rule that sets $a[\ell]$.

\subsection{Target Motion Model and Extended Kalman Filtering}\label{sec:ekf}
Let $\bfx[\ell] := [x_t[\ell], \dot x_t[\ell], y_t[\ell], \dot y_t[\ell]]^T$ collect the target position and velocity. Since the slots are short, the target velocity is approximately constant from one slot to the next. We therefore adopt the near-constant-velocity model \cite{BarShalom2001}, given by
\begin{equation}\label{eq:dyn}
\bfx[\ell] = \bfPhi \bfx[\ell-1] + \bfn[\ell-1], \qquad \bfn \sim \mathcal N(\mathbf 0, \bfQ_n),
\end{equation}
where the transition matrix and the process-noise covariance are
\begin{equation}
\bfPhi = \begin{bmatrix} 1 & \delta & 0 & 0\\ 0&1&0&0\\ 0&0&1&\delta\\ 0&0&0&1 \end{bmatrix}\!, \quad
\bfQ_n = \mathrm{diag}\big(\sigma_{n_x}^2, \sigma_{n_{v_x}}^2, \sigma_{n_y}^2, \sigma_{n_{v_y}}^2\big).
\end{equation}
Here $\delta$ is the TS duration introduced in Section~\ref{sec:model}, and the zero-mean process noise $\bfn = [n_x, n_{v_x}, n_y, n_{v_y}]^T$ absorbs small deviations from a constant velocity, such as braking, accelerating, and turning. The measurement map $h(\bfx)$ of \eqref{eq:meas} expresses the observed range, radial velocity, and azimuth as functions of the state. The three components of $h$ are
\begin{subequations}\label{eq:hmap}
\begin{align}
r(\bfx) &= \sqrt{H^2 + x_t^2 + y_t^2}, \label{eq:hr}\\
\dot r(\bfx) &= \frac{x_t \dot x_t + y_t \dot y_t}{r(\bfx)}, \label{eq:hrdot}\\
\varphi(\bfx) &= \mathrm{atan2}(y_t, x_t), \label{eq:hphi}
\end{align}
\end{subequations}
where, at the true state, the UAV-target distance $r(\bfx)$ coincides with $d_{ut}[\ell]$ of Section~\ref{sec:echomodel}, $\dot r(\bfx)$ is observed through the Doppler shift, and $\varphi(\bfx)$ is obtained from the arrival direction of the echo.

The filtering problem defined by \eqref{eq:dyn} has linear dynamics but a nonlinear measurement map. The EKF \cite{BarShalom2001} handles the nonlinearity by linearizing $h$ around the current prediction and then applying the Kalman recursion. The linearization uses the Jacobian $\bfH = \partial h/\partial \bfx \in \mathbb R^{3\times 4}$, whose rows are
\begin{IEEEeqnarray}{rCl}
\frac{\partial r}{\partial \bfx} &=& \Big[ \frac{x_t}{r},\ 0,\ \frac{y_t}{r},\ 0 \Big], \IEEEnonumber\\
\frac{\partial \dot r}{\partial \bfx} &=& \Big[ \frac{\dot x_t}{r} - \frac{x_t\, \varrho}{r^3},\ \frac{x_t}{r},\ \frac{\dot y_t}{r} - \frac{y_t\, \varrho}{r^3},\ \frac{y_t}{r} \Big], \IEEEnonumber\\
\frac{\partial \varphi}{\partial \bfx} &=& \Big[ \frac{-y_t}{x_t^2 + y_t^2},\ 0,\ \frac{x_t}{x_t^2+y_t^2},\ 0 \Big],\label{eq:jac}
\end{IEEEeqnarray}
where $\varrho := x_t \dot x_t + y_t \dot y_t$.

In each TS, the EKF proceeds in two steps. The \emph{prediction} step propagates the previous estimate and its error covariance through the motion model and computes
\begin{equation}\label{eq:ekfpred}
\hat\bfx[\ell|\ell-1] = \bfPhi \hat\bfx[\ell-1], \qquad
\bfP[\ell|\ell-1] = \bfPhi \bfP[\ell-1] \bfPhi^T + \bfQ_n.
\end{equation}
The \emph{update} step corrects the prediction using the \emph{innovation} $\bfy_s[\ell] - h(\hat\bfx[\ell|\ell-1])$, the discrepancy between the observed and the predicted measurement, and computes
\begin{IEEEeqnarray}{rCl}
\bfK[\ell] &=& \bfP[\ell|\ell-1]\bfH^T[\ell] \big(\bfH[\ell]\bfP[\ell|\ell-1]\bfH^T[\ell] + \bfR_s[\ell]\big)^{-1}\!, \IEEEeqnarraynumspace\label{eq:ekfgain}\\
\hat\bfx[\ell] &=& \hat\bfx[\ell|\ell-1] + \bfK[\ell]\big( \bfy_s[\ell] - h(\hat\bfx[\ell|\ell-1]) \big), \label{eq:ekfup}\\
\bfP[\ell] &=& \big(\bfI - \bfK[\ell]\bfH[\ell]\big) \bfP[\ell|\ell-1]. \label{eq:ekfcov}
\end{IEEEeqnarray}
Here the \emph{Kalman gain} $\bfK[\ell]$ balances the prediction uncertainty $\bfP[\ell|\ell-1]$ against the measurement noise $\bfR_s[\ell]$, $\bfH[\ell]$ denotes the Jacobian \eqref{eq:jac} evaluated at the predicted state, and the azimuth component of the innovation is wrapped to $(-\pi,\pi]$. The trigger of Section~\ref{sec:trigger} operates on the predicted state $\hat\bfx[\ell|\ell-1]$, which provides the expected target position \emph{before} the slot is transmitted.

\subsection{Posterior Cram\'er-Rao Bound}\label{sec:pcrb}
The EKF covariance $\bfP[\ell]$ assesses the filter's accuracy through the linearized model, and any underestimation of the error carries over to the trigger. We therefore quantify the estimation accuracy by the PCRB \cite{Tichavsky1998}, which holds for every estimator and is independent of the filter in use.
For the filtering problem \eqref{eq:dyn}, the PCRB states that the error covariance of \emph{any} estimator $\hat\bfx[\ell]$ of the random state $\bfx[\ell]$ based on the observations $\bfy_s[1{:}\ell]$ satisfies
\begin{equation}\label{eq:pcrbdef}
\mathbb E\big[ (\hat\bfx[\ell]-\bfx[\ell])(\hat\bfx[\ell]-\bfx[\ell])^T \big] \succeq \bfJ[\ell]^{-1},
\end{equation}
where $\bfJ[\ell]$ is the Bayesian Fisher information matrix of the problem and $\succeq$ denotes the ordering of positive semidefinite matrices. By \cite{Tichavsky1998}, $\bfJ[\ell]$ follows a recursion over $\ell$, which for our linear dynamics with Gaussian noise takes the form
\begin{IEEEeqnarray}{rCl}
\bfJ[\ell|\ell-1] &=& \big( \bfQ_n + \bfPhi \bfJ[\ell-1]^{-1} \bfPhi^T \big)^{-1}, \label{eq:pcrbpred}\\
\bfJ[\ell] &=& \bfJ[\ell|\ell-1] + \mathbb E\big[ \bfH^T[\ell] \bfR_s[\ell]^{-1} \bfH[\ell] \big], \label{eq:pcrbup}
\end{IEEEeqnarray}
initialized at $\bfJ[0] = \bfP[0]^{-1}$. The prediction step \eqref{eq:pcrbpred} propagates the information through the motion model and can only decrease it, and the update step \eqref{eq:pcrbup} adds the information in the new observation.

The expectation in \eqref{eq:pcrbup} is taken over the random target trajectories through the state-dependent Jacobian $\bfH[\ell]$ and has no closed form. Following \cite{Tichavsky1998, Jiang2024}, we evaluate this term at the realized trajectory.%

Finally, we reduce the matrix bound to a scalar. The trigger requires only the \emph{position} accuracy, and the diagonal entries of $\bfJ[\ell]^{-1}$ bound the mean-square errors of the individual state components. Summing the two entries associated with $x_t$ and $y_t$ therefore yields the position bound
\begin{IEEEeqnarray}{rCl}\label{eq:pcrbpos}
\mathrm{PCRB}_{\mathrm{pos}}[\ell] &:=& \mathrm{tr}\big([\bfJ[\ell]^{-1}]_{\mathrm{pos}}\big) \IEEEnonumber\\
&\le& \mathbb E\big[ (\hat x_t[\ell]-x_t[\ell])^2 + (\hat y_t[\ell]-y_t[\ell])^2 \big],
\end{IEEEeqnarray}
where $[\cdot]_{\mathrm{pos}}$ denotes the $2\times 2$ submatrix built from the rows and columns of $x_t$ and $y_t$. The scalar $\mathrm{PCRB}_{\mathrm{pos}}[\ell]$ is the sensing metric of the setup. The power shares $(\lambda, \gamma)$ and the beam split $\omega$ shape the beampattern \eqref{eq:beampattern}, the beampattern sets the echo SNR \eqref{eq:snrecho}, and the echo SNR scales the measurement covariance $\bfR_s[\ell]$ that enters the information update \eqref{eq:pcrbup}. This chain is how the transmit design enters the sensing constraint \eqref{eq:optsense}.

\subsection{PCRB-Based Trigger}\label{sec:trigger}
The aperiodic traffic should be triggered when the target is inside $\mathcal R$. Since the true position is unknown, we base the trigger on the \emph{predicted belief}. At TS $\ell$, the target position is modeled as Gaussian with mean $\bfmu_\ell = [\hat x_t[\ell|\ell-1], \hat y_t[\ell|\ell-1]]^T$ and covariance $\bfSig_\ell = [\bfJ[\ell|\ell-1]^{-1}]_{\mathrm{pos}}$, i.e., the \emph{predicted PCRB}.
 The trigger probability is defined as
\begin{equation}\label{eq:ptrig}
p_{\mathrm{trigger}}[\ell] := \Pr\big[ \mathbf p \in \mathcal R \big] \quad \text{with } \mathbf p \sim \mathcal N(\bfmu_\ell, \bfSig_\ell).
\end{equation}
The trigger state is
\begin{equation}\label{eq:a}
a[\ell] := \indicator{p_{\mathrm{trigger}}[\ell] \ge \rho},
\end{equation}
where $\indicator{\cdot}$ denotes the indicator function and $\rho \in [0,1]$ is a design threshold. The aperiodic stream is thus activated when the belief assigns probability at least $\rho$ to the target being inside $\mathcal R$.

The probability \eqref{eq:ptrig} is the integral of a bivariate Gaussian distribution over an offset disc and has no closed form. The classical treatment expands this integral into infinite series and truncates them \cite{DiDonato1961}, and Monte Carlo integration converges only as the inverse square root of the number of samples and adds estimation randomness to the threshold decision \eqref{eq:a}. Proposition~\ref{prop:trigger} expresses the same probability as a single integral over the angle, in which the radial direction is integrated exactly and only the angular integral is evaluated numerically.

\begin{proposition}\label{prop:trigger}
Let $\mathbf p \sim \mathcal N(\bfmu, \bfSig)$ in $\mathbb R^2$ with $\bfSig$ positive definite, let $\mathcal D = \{\mathbf p : \|\mathbf p - \mathbf c\| \le R\}$ be a disc, and let $\bfSig = \mathbf L \mathbf L^T$ with $\mathbf L$ the Cholesky factor. Write $\mathbf m := \bfmu - \mathbf c$ and, for each angle $\vartheta \in [0, 2\pi)$, let $\boldsymbol\eta(\vartheta) := [\cos\vartheta, \sin\vartheta]^T$, and define
\begin{IEEEeqnarray}{rCl}\label{eq:quadcoef}
a(\vartheta) &:=& \|\mathbf L \boldsymbol\eta(\vartheta)\|^2,  \;
b(\vartheta) := 2\, \mathbf m^T \mathbf L \boldsymbol\eta(\vartheta), \;
c_0 := \|\mathbf m\|^2 - R^2. \notag \\
\end{IEEEeqnarray}
If $b(\vartheta)^2 \ge 4 a(\vartheta) c_0$, let $r_\pm(\vartheta) = \big(-b(\vartheta) \pm \sqrt{b(\vartheta)^2 - 4 a(\vartheta) c_0}\big)/(2 a(\vartheta))$ and set $r_{\mathrm{lo}}(\vartheta) = \max\{r_-(\vartheta), 0\}$ and $r_{\mathrm{hi}}(\vartheta) = \max\{r_+(\vartheta), 0\}$; otherwise set $r_{\mathrm{lo}}(\vartheta) = r_{\mathrm{hi}}(\vartheta) = 0$. Then
\begin{equation}\label{eq:proptrig}
\Pr[\mathbf p \in \mathcal D] = \frac{1}{2\pi} \int_0^{2\pi} \Big( e^{-r_{\mathrm{lo}}^2(\vartheta)/2} - e^{-r_{\mathrm{hi}}^2(\vartheta)/2} \Big)\, d\vartheta.
\end{equation}
\end{proposition}
\begin{IEEEproof}
First, we whiten the Gaussian. Writing $\mathbf p = \bfmu + \mathbf L \boldsymbol\xi$ with $\boldsymbol\xi \sim \mathcal N(\mathbf 0, \bfI_2)$, the event $\mathbf p \in \mathcal D$ becomes $\|\mathbf m + \mathbf L \boldsymbol\xi\| \le R$. Second, we pass to polar coordinates $\boldsymbol\xi = r \boldsymbol\eta(\vartheta)$ with $r \ge 0$ and $\vartheta \in [0, 2\pi)$. Under the standard two-dimensional Gaussian, the angle $\vartheta$ is uniform, the radius $r$ has density $r e^{-r^2/2}$, and the two are independent. The probability therefore becomes
\begin{equation}\label{eq:proofpolar}
\Pr[\mathbf p \in \mathcal D] = \frac{1}{2\pi} \int_0^{2\pi} \int_{\{r \ge 0 :\, q_\vartheta(r) \le 0\}} r e^{-r^2/2}\, dr\, d\vartheta,
\end{equation}
where $q_\vartheta(r) := \|\mathbf m + r \mathbf L\boldsymbol\eta(\vartheta)\|^2 - R^2 = a(\vartheta) r^2 + b(\vartheta) r + c_0$ with the coefficients \eqref{eq:quadcoef}. Third, we identify the inner integration set. Since $\bfSig$ is positive definite, $a(\vartheta) > 0$, and the quadratic $q_\vartheta$ is nonpositive on the interval $[r_-(\vartheta), r_+(\vartheta)]$ when its roots are real, and nowhere otherwise. Intersecting this interval with $r \ge 0$ gives $[r_{\mathrm{lo}}(\vartheta), r_{\mathrm{hi}}(\vartheta)]$ as defined in the statement; the empty cases are covered by $r_{\mathrm{lo}} = r_{\mathrm{hi}}$, for which the inner integral vanishes. Finally, the inner integral evaluates to
\begin{equation}
\int_{r_{\mathrm{lo}}}^{r_{\mathrm{hi}}} r e^{-r^2/2}\, dr = e^{-r_{\mathrm{lo}}^2/2} - e^{-r_{\mathrm{hi}}^2/2}.
\end{equation}
Substituting this expression into \eqref{eq:proofpolar} yields \eqref{eq:proptrig}.
\end{IEEEproof}
The integrand of \eqref{eq:proptrig} is a continuous $2\pi$-periodic function of $\vartheta$. The integral is therefore evaluated by trapezoidal quadrature with a few hundred nodes at negligible cost. Since the radial integral is evaluated exactly, the quadrature stays accurate even for small covariances $\bfSig_\ell$. The trigger itself is not a variable of problem \eqref{eq:opt}. It connects to the design through $a[\ell]$, which defines the hypotheses \eqref{eq:evetest} of the covertness constraint \eqref{eq:optcov}, and through $\bfSig_\ell$, which the sensing constraint \eqref{eq:optsense} bounds.

\section{Keyless Diffusion-Based Covert Transmission}\label{sec:scheme}
This section completes the security layer. Section~\ref{sec:ddim} constructs the diffusion sampler of the safeguard, and Section~\ref{sec:receivers} develops the receivers with which Bob decodes without shared randomness. The sampler provides the optimization variable $\theta$ of problem \eqref{eq:opt}, and the error probability of the receivers is the aperiodic reliability $\epsilon_a$ in \eqref{eq:optrel}.

\subsection{DDIM Safeguard Generator}\label{sec:ddim}
At every TS, the safeguard requires a generator that produces a new random waveform with a trainable distribution.
 A safeguard drawn from a fixed distribution, such as matched-power Gaussian noise, spreads its randomness uniformly over the signal space, and hiding the trigger then falls on the placement of the codewords alone. The diffusion model instead makes the safeguard distribution a trainable variable of the distribution-matching objective that covertness demands, and training it jointly with the codebook shapes the masking around the codeword geometry it must hide. Diffusion models turn Gaussian noise into a sample of a learned distribution through a sequence of neural denoising steps, with the denoiser parameters $\theta$ as the trainable degrees of freedom. We use the DDIM \cite{Song2021}, a variant of the denoising diffusion probabilistic model \cite{Ho2020} whose sampler needs only a few steps.

A diffusion model is specified through a \emph{forward} process that turns a clean waveform $\bfw_0 \in \mathbb C^{1\times n}$ into noise. With $\bfw_t$ denoting the chain state at diffusion step $t \in \{0, 1, \ldots, T\}$, the forward process is defined by its Gaussian marginals
\begin{equation}\label{eq:ddimforward}
q(\bfw_t | \bfw_0) = \mathcal N_{\mathbb C}\big(\sqrt{\alpha_t}\, \bfw_0,\ (1-\alpha_t)\, \bfI\big).
\end{equation}
The decreasing sequence $\alpha_t \in (0,1]$ sets the proportion of waveform and noise at step $t$. With $\alpha_0 \approx 1$ the chain starts at the clean waveform, and with $\alpha_T \approx 0$ the chain ends in pure noise.
Generation runs this corruption backward. The only learned component is a neural network $\bfeps_\theta(\bfw_t, t \,|\, \psi)$ that estimates the noise contained in the state $\bfw_t$, conditioned on the step index and the public tag. Writing \eqref{eq:ddimforward} as $\bfw_t = \sqrt{\alpha_t}\, \bfw_0 + \sqrt{1-\alpha_t}\, \bfeps$ and replacing the standard Gaussian $\bfeps$ by its estimate gives the network's estimate of the clean waveform,
\begin{equation}\label{eq:w0hat}
\hat \bfw_0 = \frac{\bfw_t - \sqrt{1-\alpha_t}\; \bfeps_\theta(\bfw_t, t \,|\, \psi)}{\sqrt{\alpha_t}}.
\end{equation}
The DDIM reverse update is
\begin{equation}\label{eq:ddimstep}
\bfw_{t-1} = \sqrt{\alpha_{t-1}}\, \hat \bfw_0 + \sqrt{1 - \alpha_{t-1} - \sigma_t^2}\; \bfeps_\theta(\bfw_t, t \,|\, \psi) + \sigma_t \bfxi_t,
\end{equation}
where $\bfxi_t \sim \mathcal N_{\mathbb C}(\mathbf 0, \bfI)$ is independent noise. The first term places the clean estimate at the signal scale of step $t-1$, the second scales the estimated noise such that the state matches the forward marginal \eqref{eq:ddimforward} at step $t-1$, and the third injects new randomness with standard deviation
\begin{equation}\label{eq:eta}
\sigma_t = \eta \sqrt{\frac{1-\alpha_{t-1}}{1-\alpha_t}} \sqrt{1 - \frac{\alpha_t}{\alpha_{t-1}}}, \qquad \eta \ge 0,
\end{equation}
controlled by the stochasticity parameter $\eta$. With $\eta = 0$ the update is deterministic given the current state, and $\eta = 1$ recovers the sampling of the original diffusion model \cite{Ho2020}. The update \eqref{eq:ddimstep} remains valid between non-consecutive steps, and the reverse process may therefore jump along a grid
of only $K \ll T$ points \cite{Song2021}.

The safeguard of TS $\ell$ is obtained as follows. The UAV draws an initial state $\bfw_{t_K} \sim \mathcal N_{\mathbb C}(\mathbf 0, \bfI)$ and runs the $K$ reverse steps \eqref{eq:ddimstep} along the grid, conditioned on the tag $\psi_\ell$. After each step the state is rescaled to unit average power, a final projector enforces the output power exactly, and $\bfW_{\theta,\ell}$ is set to the final state $\bfw_{t_0}$. The per-slot cost is $K$ network evaluations.

\subsection{Seed-Free Receivers at Bob}\label{sec:receivers}
The aperiodic message carries a time-critical alert and must be decoded within its own TS. Bob therefore applies successive interference cancellation in every TS. Bob first decodes the periodic message while the aperiodic stream is treated as additional noise, cancels the periodic component at the known gain $\beta_p$, and then decodes the aperiodic message from the remaining signal. The periodic stage sees the per-symbol signal power $\lambda\P|\beta_p|^2$ against the interference-plus-noise power $\sigma_v^2 := (1-\lambda)\gamma\P + \sigma_b^2/(N|g_b|^2)$ and operates at the ratio $\mathrm{SINR}_p := \lambda\P|\beta_p|^2/\sigma_v^2$. In the single-antenna setting, where the safeguard is not nulled at Bob, its power $(1-\lambda)(1-\gamma)\P$ adds to $\sigma_v^2$. The periodic stream uses a standard short-blocklength code at the reliability target $\bar\epsilon_p$, and the normal approximation of the maximal coding rate over the complex Gaussian channel \cite[Th.~54, eq.~(294)]{Polyanskiy2010} sets the periodic rate to
\begin{equation}\label{eq:pacq}
R_p \approx \log_2\!\big(1 + \mathrm{SINR}_p\big) - \sqrt{V_p/n}\; Q^{-1}(\bar\epsilon_p) + \frac{\log_2 n}{2n},
\end{equation}
where $V_p := \frac{\mathrm{SINR}_p(\mathrm{SINR}_p + 2)}{(1+\mathrm{SINR}_p)^2}\, \log_2^2 e$ is the dispersion of the effective channel\footnote{An information-theoretic treatment of this superposition based on dirty-paper coding is given in \cite{Nikbakht2025}.}. When the periodic stage decodes correctly, cancellation of the periodic component and equalization by $(g_b^* \sqrt N)$ reduce \eqref{eq:ybeff} to
\begin{equation}\label{eq:residual}
\tilde \bfy_{b,\ell} \approx \sqrt{c_1}\, \bfX_{a,\ell} + \tilde \bfz_{b,\ell},
\end{equation}
with the aperiodic power $c_1 := (1-\lambda)\gamma\P$, $\tilde \bfz_{b,\ell} \sim \mathcal N_{\mathbb C}(0, \tilde\sigma_b^2 \bfI)$, and $\tilde \sigma_b^2 = \sigma_b^2/(N |g_b|^2)$. With exact zero-forcing the safeguard does not reach Bob and \eqref{eq:residual} is a standard Gaussian channel. Imperfect nulling, or the single-antenna setting, adds the safeguard term $\sqrt{c_2}\, \bfW_{\theta,\ell}$ to \eqref{eq:residual}, where $c_2 := (1-\lambda)(1-\gamma)\P$. Bob shares no seed with the UAV, cannot regenerate $\bfW_{\theta,\ell}$, and can only exploit its \emph{distribution}. The maximum-likelihood (ML) rule marginalizes over this distribution and is given by
\begin{equation}\label{eq:mlexact}
\hat m_{a,\ell} = \arg\max_m\ \mathbb E_{\bfW}\Big[ p_{\tilde \bfz}\big( \tilde \bfy_{b,\ell} - \sqrt{c_1}\, \bfX_a(m) - \sqrt{c_2}\, \bfW \big) \Big],
\end{equation}
where $p_{\tilde\bfz}$ is the Gaussian density of the residual noise. This optimum is intractable, since the diffusion output distribution has no closed-form density and the expectation requires averaging many safeguard samples per candidate and per slot. The two rules below approximate \eqref{eq:mlexact} from different directions.
With multiple transmit antennas, where the safeguard is nulled toward Bob, $c_2$ vanishes at Bob and the first rule coincides with \eqref{eq:mlexact}. The comparison therefore matters for the single-transmit-antenna and imperfect-nulling cases.

\subsubsection{Isotropic marginalized ML} Approximating the safeguard by white Gaussian noise of the same power reduces \eqref{eq:mlexact} to the classical minimum-distance rule \cite{Poor}
\begin{equation}\label{eq:mliso}
\hat m_{a,\ell} = \arg\min_m \big\| \tilde \bfy_{b,\ell} - \sqrt{c_1}\, \bfX_a(m) \big\|^2 .
\end{equation}
This rule needs no training, costs one correlation per candidate, and is optimal when $c_2 = 0$.

\subsubsection{Trained denoiser decoder} The second rule uses a denoiser that Bob trains for decoding, following diffusion-based solvers for inverse problems \cite{Chung2023}. The generator's denoiser $\bfeps_\theta$ is trained only through the sampler and is not trained to remove noise from a finished safeguard. Bob, however, knows $\theta$ and the tag and can therefore draw any number of safeguard samples offline. With these samples, Bob trains a network $\bfeps_\nu(\bfr, s)$ by denoising score matching \cite{Vincent2011},
\begin{equation}\label{eq:dsm}
\min_\nu\ \mathbb E_{\bfW, s, \bfxi}\Big[ \big\| \bfeps_\nu\big(\bfW + s\,\bfxi,\, s\big) - \bfxi \big\|^2 \Big],
\end{equation}
where $\bfW$ is a public safeguard sample, $\bfxi \sim \mathcal N_{\mathbb C}(\mathbf 0, \bfI)$, and the noise level $s$ is drawn log-uniformly from $[s_b/2,\, 4 s_b]$ with $s_b := \tilde\sigma_b/\sqrt{c_2}$. For each candidate $m$, Bob forms the normalized residual $\bfr_m = (\tilde \bfy_{b,\ell} - \sqrt{c_1} \bfX_a(m))/\sqrt{c_2}$, which equals $\bfW_{\theta,\ell} + s_b \bfxi$ for the transmitted candidate. By Tweedie's formula \cite{Efron2011}, the trained network yields the posterior mean of the safeguard, $\hat\bfW_m = \bfr_m - s_b\, \bfeps_\nu(\bfr_m, s_b)$. The rule decides for the candidate with the smallest reconstruction error,
\begin{equation}\label{eq:dendec}
\hat m_{a,\ell} = \arg\min_m \big\| \bfr_m - \hat\bfW_m \big\|^2 = \arg\min_m\ s_b^2 \big\| \bfeps_\nu(\bfr_m, s_b) \big\|^2.
\end{equation}
The reconstruction error is small only when $\bfr_m$ lies close to the safeguard distribution, which is the case for the transmitted candidate. The network is trained once per operating point before deployment.

\section{End-to-End Training}\label{sec:objective}\label{sec:training}
In this section, we present the training procedure that solves the design problem \eqref{eq:opt} of Section~\ref{sec:problem}. Problem \eqref{eq:opt} cannot be solved with standard constrained optimization for two reasons. First, the covertness and reliability constraints depend on the sampler weights $\theta$ only through the distribution of the safeguard that the network generates. This distribution has no closed-form expression, and the constraints can therefore only be estimated by drawing samples from the network. Second, the objective is the rate of an explicit finite codebook and changes only with the integer $M_a$.
We therefore hold $(M_a, \lambda, \gamma, \omega)$ fixed during training and sweep the four parameters over a grid across training runs, which traces the trade-offs of Section~\ref{sec:numerical}.
For each setting, the continuous variables $(\theta, \zeta, \bfX_a)$ are trained by turning the constraints into penalty terms of a single loss \cite{Boyd2004}. Three constraints, \eqref{eq:optpow}, \eqref{eq:optsense}, and the periodic part of \eqref{eq:optrel}, need no penalty. The power constraint \eqref{eq:optpow} is enforced exactly by the projectors of Section~\ref{sec:ddim}. The sensing constraint \eqref{eq:optsense} depends only on $(\lambda, \gamma, \omega)$. The periodic part of \eqref{eq:optrel} holds by construction, since for each $(\lambda, \gamma)$ the periodic rate is set by \eqref{eq:pacq} at the target $\bar\epsilon_p$. Differentiable penalties are needed for the covertness constraint \eqref{eq:optcov} and the aperiodic part of \eqref{eq:optrel}.

For covertness, the total variation distance itself cannot serve as a training signal. Evaluating the supremum in \eqref{eq:tvdef} requires the densities of Eve's observations under the two hypotheses, and the network that generates the safeguard provides samples of these distributions but not their densities. Training requires a loss computed from sample batches and differentiated with respect to $\theta$, and we use two complementary losses of this kind. The first is the squared maximum mean discrepancy (MMD) \cite{Gretton2012}, defined for a kernel $k$ as
\begin{IEEEeqnarray}{rCl}\label{eq:mmd}
\mathrm{MMD}^2\big(P^{(0)}_\psi, P^{(1)}_\psi\big) &=& \mathbb E[k(y,y')] + \mathbb E[k(w,w')] \IEEEnonumber\\
&&\; -\, 2\mathbb E[k(y,w)],
\end{IEEEeqnarray}
with $y, y' \sim P^{(0)}_\psi$ and $w, w' \sim P^{(1)}_\psi$ independent. The three expectations in \eqref{eq:mmd} can be estimated from two batches of generated samples, and the resulting estimate is differentiable with respect to $\theta$.
 The second loss is \emph{adversarial} \cite{Goodfellow2014}. A classifier $C_\phi$ is trained to detect the trigger state from Eve's observations, and the generator is trained to make the classifier posterior uniform. Since \eqref{eq:tvbound} holds with equality for the likelihood-ratio detector, the best detector's accuracy equals $\frac{1}{2}(1 + \dtv)$, and lowering the accuracy of a strong classifier lowers an empirical estimate of $\dtv$.%

For the aperiodic stream, the error event behind $\epsilon_a$ is a binary indicator and gives no gradient. We therefore replace $\epsilon_a$ during training by the cross-entropy (CE) of a decoder network $D_\zeta(\tilde \bfy_{b,\ell}, \psi_\ell)$, trained jointly with the safeguard generator to output the posterior probabilities of the candidate messages as in learned demappers \cite{OShea2017}. This differentiable loss equals the negative log-probability that the network assigns to the transmitted index, and lowering the loss lowers $\epsilon_a$. With all penalty terms defined, each training iteration alternates the classifier step
\begin{equation}\label{eq:advstep}
\min_\phi\; \mathrm{CE}\big( C_\phi(\bfy_{e,\ell}), a[\ell] \big)
\end{equation}
with the generator step
\begin{IEEEeqnarray}{rCl}\label{eq:loss}
\min_{\theta, \zeta, \bfX_a}&& \lambda_{\mathrm{priv}}\, \mathrm{MMD}^2\big(\{\bfy_{e,\ell} | a{=}0\}, \{\bfy_{e,\ell} | a{=}1\}\big) \IEEEnonumber\\
&& +\, \lambda_{\mathrm{adv}}\, \mathbb E\Big[ \textstyle\sum_{a'} q_\phi(a' | \bfy_{e,\ell}) \log q_\phi(a' | \bfy_{e,\ell}) \Big] \IEEEnonumber\\
&& +\, \lambda_{\mathrm{comm}}\, \mathrm{CE}\big( D_\zeta(\tilde \bfy_{b,\ell}, \psi_\ell), \tilde m_{a,\ell} \big) + \lambda_{\mathrm{rf}} \mathcal L_{\mathrm{RF}} \IEEEnonumber\\
&& +\, \lambda_{\mathrm{mean}}\, \mathcal L_{\mathrm{mean}} + \lambda_{\mathrm{spr}}\, \mathcal L_{\mathrm{spr}},\IEEEeqnarraynumspace
\end{IEEEeqnarray}
where $q_\phi(\cdot | \bfy_{e,\ell})$ is the classifier posterior over the trigger state. The second term is smallest when this posterior is uniform and the classifier learns nothing from Eve's observation. The term $\mathcal L_{\mathrm{RF}}$, weighted by $\lambda_{\mathrm{rf}} > 0$, penalizes safeguard samples whose power deviates from the budget. The projectors of Section~\ref{sec:ddim} enforce the budget exactly, and this penalty keeps the network output close to its projection and thereby stabilizes training. The first two terms penalize the covertness constraint \eqref{eq:optcov} and the third the aperiodic part of the reliability constraint \eqref{eq:optrel}. The weights $\lambda_{\mathrm{priv}}, \lambda_{\mathrm{adv}}, \lambda_{\mathrm{comm}} > 0$ balance the covertness and reliability penalties and select the operating point of the design.

The last two terms keep the safeguard random.  For a batch of $N_s$ safeguard samples $\bfW_1, \ldots, \bfW_{N_s}$, the first term $\mathcal L_{\mathrm{mean}}$ is the energy share of the batch mean and is defined by
\begin{equation}\label{eq:lmean}
\mathcal L_{\mathrm{mean}} := \frac{\big\| \frac{1}{N_s}\sum_{i=1}^{N_s} \bfW_i \big\|^2}{\frac{1}{N_s}\sum_{i=1}^{N_s} \| \bfW_i \|^2},
\end{equation}
which reaches zero when no common waveform dominates the samples. The second term $\mathcal L_{\mathrm{spr}} $ rewards the spread of the samples through the nearest-neighbor distance and is given by
\begin{equation}\label{eq:lspr}
\mathcal L_{\mathrm{spr}} := -\frac{1}{N_s} \sum_{i=1}^{N_s} \log \min_{j \neq i} \| \bfW_i - \bfW_j \|^2,
\end{equation}
which up to a positive scale and an additive constant is the negative of the nearest-neighbor estimate of the safeguard entropy \cite{Kozachenko1987}. The weights $\lambda_{\mathrm{mean}}, \lambda_{\mathrm{spr}} > 0$ set the strength of the two terms.

The aperiodic codebook $\bfX_a(\cdot)$ is optimized jointly with the sampler and remains public, and projection after every update keeps its codewords power-equalized. Each training batch samples tags, channels, and trigger states, and applies classifier-free conditioning on the tag \cite{HoSalimans2022}. Every safeguard sample is produced by $K$ DDIM steps, the loss gradient passes backward through all $K$ steps, and exponential moving average weights with gradient clipping stabilize this differentiation \cite{Ho2020}. By default, training evaluates Eve's terms at the nominal direction $\hat u_e$. A robust variant instead draws $\delta_e$ uniformly from $[-\Delta_e, \Delta_e]$ at every training slot, which enforces the covertness terms in \eqref{eq:loss} on average over the uncertainty set.

The training problem is nonconvex, and training returns a stationary point of the loss \eqref{eq:loss} rather than a global optimum of problem \eqref{eq:opt}. Therefore, we establish the feasibility of the trained design after training. The power constraint \eqref{eq:optpow} holds by construction. The sensing constraint \eqref{eq:optsense} is evaluated directly through the PCRB recursion, and the reliability constraint \eqref{eq:optrel} is measured as Bob's error rate. The loss enforces the covertness constraint \eqref{eq:optcov} through the MMD and classifier terms rather than through the total variation distance itself. We therefore verify covertness with new detectors that are trained independently of the classifier of \eqref{eq:advstep}. By the equality form of \eqref{eq:tvbound}, a detector accuracy $\beta_{\mathrm{det}}$ implies $\dtv \ge 2\beta_{\mathrm{det}} - 1$, and a detector can expose leakage but never prove its absence. If the MMD reaches zero, the two distributions coincide and \eqref{eq:optcov} holds with $\epsilon_E = 0$.

\section{Numerical Results}\label{sec:numerical}
We evaluate the proposed system with the following parameters, which remain fixed throughout. The UAV hovers at $H = 100$~m above the origin with an $N = 8$ element ULA; Bob is at $(40, -30)$~m, and Eve at $(55, -38)$~m.
The target starts at $(100, 20)$~m with initial velocity $(\dot x_t, \dot y_t) = (-5, -3.5)$~m/s and crosses the protected region $\mathcal R$ of radius $R_0 = 40$~m centered at Bob. The flight has $L = 200$ TSs of $\delta = 0.1$~s, with $\sigma_{n_x} = \sigma_{n_y} = 0.1$~m, $\sigma_{n_{v_x}} = \sigma_{n_{v_y}} = 0.05$~m/s, $\sigma_{r,0} = 0.5$~m, $\sigma_{\dot r, 0} = 0.3$~m/s, and $\sigma_{\varphi,0} = 1^\circ$ at $\varsigma = 30$~dB and $d_{\mathrm{ref}} = 150$~m in \eqref{eq:R}, and $\rho = 0.9$. Each TS carries $n = 64$ symbols, and the power budget is normalized to $\P = 1$, with all SNRs defined relative to the normalization. The region beam $\bfw_{\mathcal R}$ is synthesized over the region's direction-cosine sector $[0, 0.61]$ and keeps a near-uniform power gain between $0.22$ and $0.39$ across the sector, with $|\beta_p|^2 = 0.34$ at Bob and $|\alpha_p|^2 = 0.38$ at Eve. The reliability target for the periodic stream is $\bar\epsilon_p = 10^{-4}$. The periodic stage operates at $\mathrm{SINR}_p = 0.68$, where \eqref{eq:pacq} gives $R_p \approx 0.26$ bit per symbol, and the periodic stream carries $16$ information bits per TS. The computed rate is achievable in practice, since cyclic-redundancy-check-aided short-blocklength codes operate within a fraction of a decibel of \eqref{eq:pacq} at blocklengths of this order \cite{Yang2022}. The denoiser $\bfeps_\theta$ is a three-block one-dimensional convolutional network with feature-wise linear modulation (FiLM) tag conditioning \cite{Perez2018}. Sampling uses $K = 6$ DDIM steps, and training minimizes the loss \eqref{eq:loss} with a Gaussian-kernel mixture for the MMD. Covertness is measured by the accuracy of a post-hoc neural detector at Eve with $12{,}000$ training and $8{,}000$ test TSs for \eqref{eq:evetest}. An accuracy of $0.5$ is a blind guess, and the same detector budget applies to every scheme and figure. We call this the \emph{learned Eve}. We also consider a \emph{model-aware Eve}, which knows the codebook, the power shares, the beam gains and the public model and applies the likelihood-ratio test for \eqref{eq:evetest}, with the safeguard likelihood estimated from $40{,}000$ public safeguard samples and the accuracy measured on $4{,}000$ test TSs. Both accuracies are lower bounds on $(1 + \dtv)/2$ in \eqref{eq:everror}, and the model-aware Eve is the tighter one. Reliability is Bob's aperiodic block error rate (BLER) over the triggered TSs under the trained denoiser decoder \eqref{eq:dendec}. Bob's network $\bfeps_\nu$ has the architecture of $\bfeps_\theta$ with four blocks of twice the width, conditioned on the noise level, and is trained on $60{,}000$ public safeguard samples for $4{,}000$ iterations. The randomness weights are $\lambda_{\mathrm{mean}} = 4$ and $\lambda_{\mathrm{spr}} = 0.05$.
The design variables default to $\lambda = 0.6$, $\gamma = 0.5$, $\omega = 0.5$, $M_a = 4$, $\eta = 1$, and Bob and Eve SNRs of $10$ and $5$~dB.

\subsection{Tracking, PCRB, and Triggering}\label{sec:numtracking}
Fig.~\ref{fig:trackingJ} verifies the tracking and the trigger on one flight realization. In Fig.~\ref{fig:trackingJ}(a) the EKF estimate lies on top of the true trajectory, with a position root mean square error (RMSE) of $0.69$~m over the flight. In Fig.~\ref{fig:trackingJ}(b) the trigger state coincides with the gray shading over the full flight. The exceptions are the two boundary crossings, where the trigger turns on two TSs after entry and off one TS before exit, since the predicted belief then lies partly inside and partly outside $\mathcal R$. The target is inside $\mathcal R$ for $123$ TSs. The trigger is active in $120$ of the $123$ in-region TSs and never activates outside.

\begin{figure}[t]
\centering
\subfloat[Geometry and estimated trajectory.]{\includegraphics[width=0.49\columnwidth]{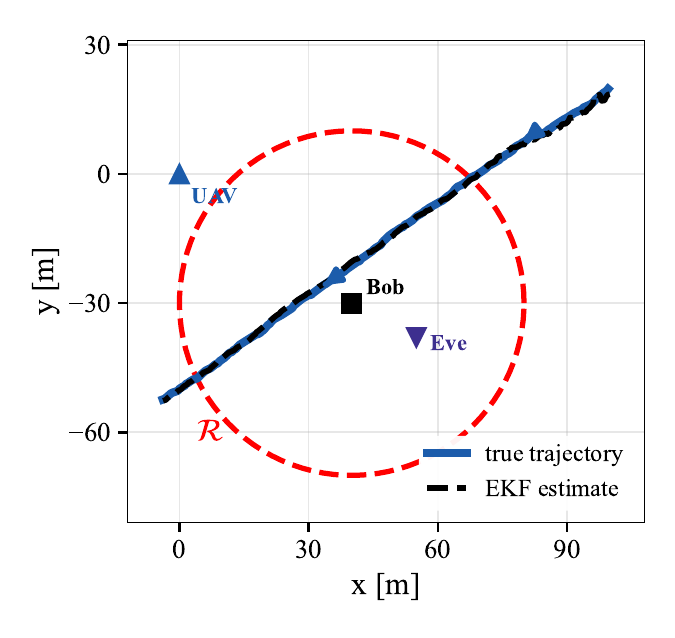}}\hfill
\subfloat[Trigger probability vs. time.]{\includegraphics[width=0.49\columnwidth]{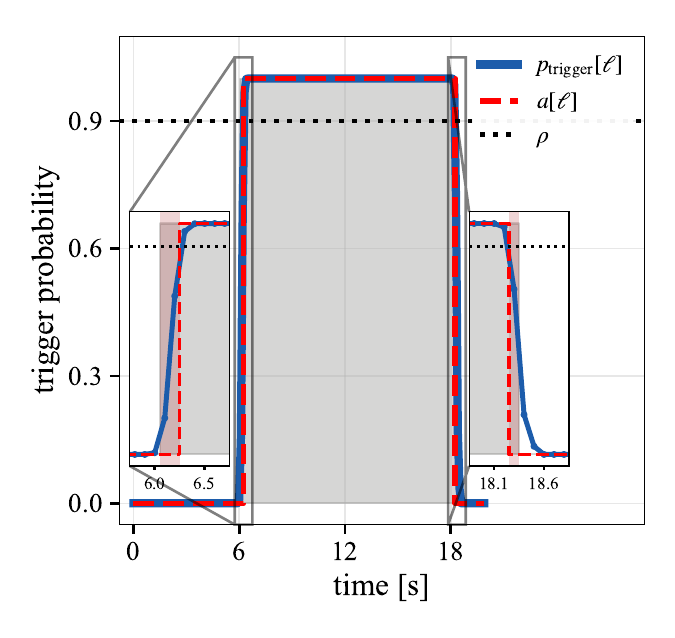}}
\caption{Tracking and triggering ($N = 8$). In (a), the arrows indicate the direction of target motion. In (b), the gray shading marks the TSs with the target truly inside $\mathcal R$, and the insets magnify the two boundary crossings, with red bands at the three TSs in which the trigger differs from the ground truth.}
\label{fig:trackingJ}
\end{figure}

\begin{figure}[t]
\centering
\subfloat[Geometry; Eve$_1$ inside the sector, Eve$_2$ outside.]{\includegraphics[width=0.7\columnwidth]{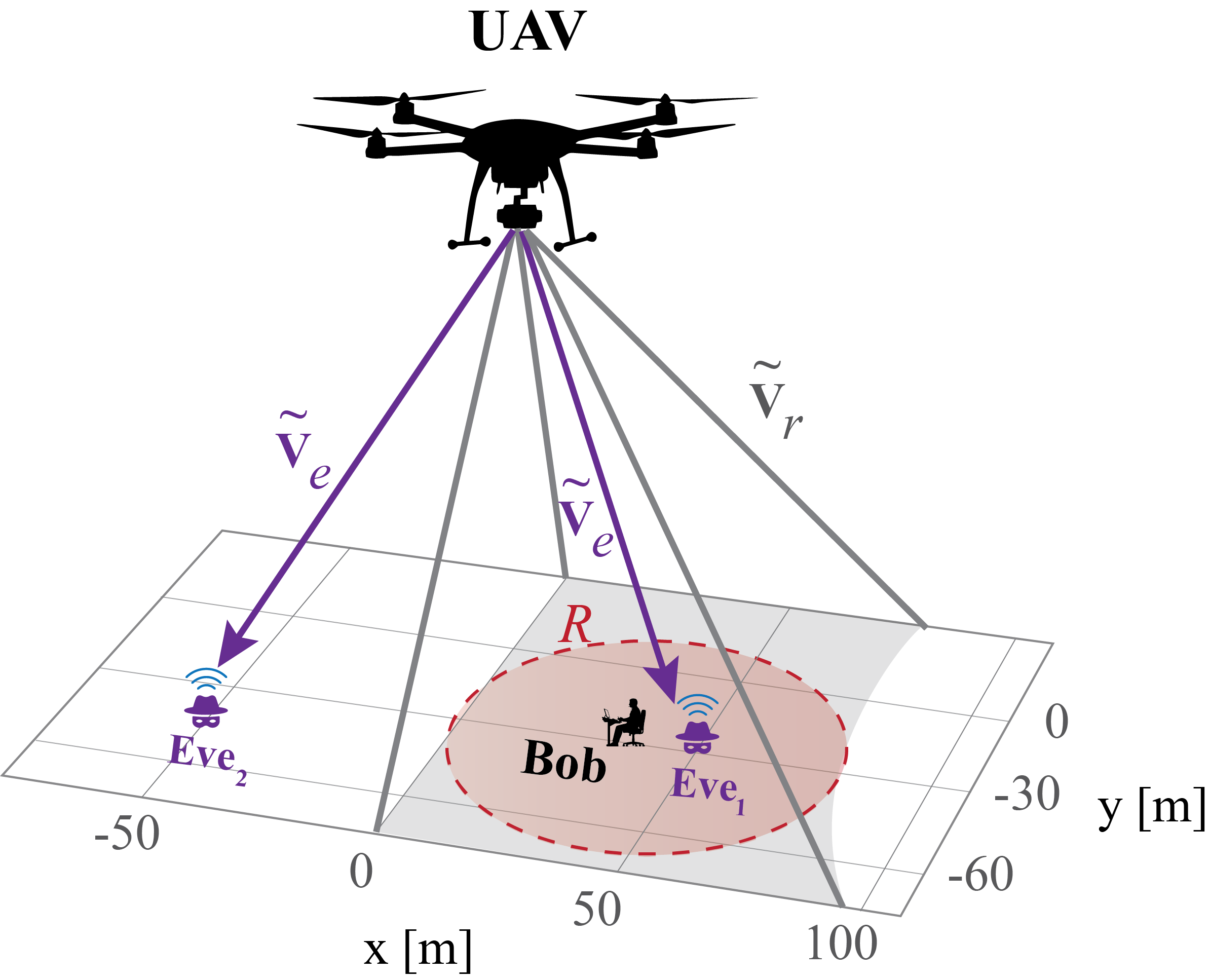}}\\
\subfloat[Eve accuracy vs. $\omega$ (Eve$_1$).]{\includegraphics[width=0.49\columnwidth]{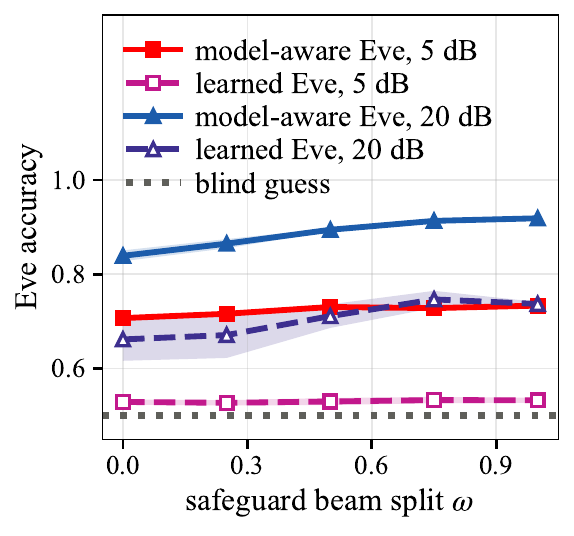}}\hfill
\subfloat[Position PCRB vs. $\omega$ (Eve$_1$).]{\includegraphics[width=0.49\columnwidth]{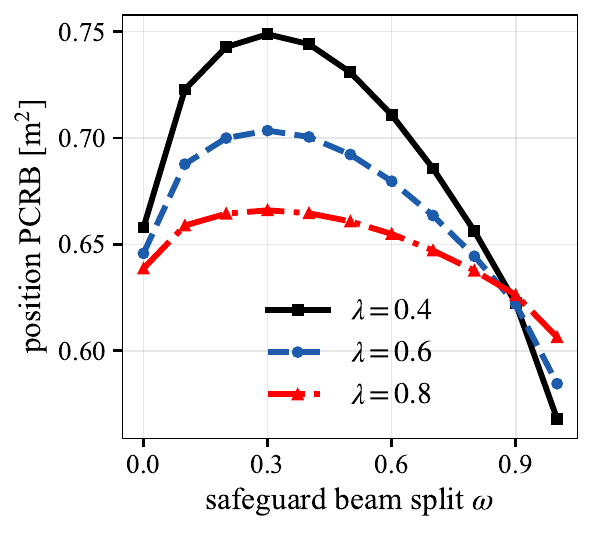}}\\
\subfloat[Eve accuracy vs. $\omega$ (Eve$_2$).]{\includegraphics[width=0.49\columnwidth]{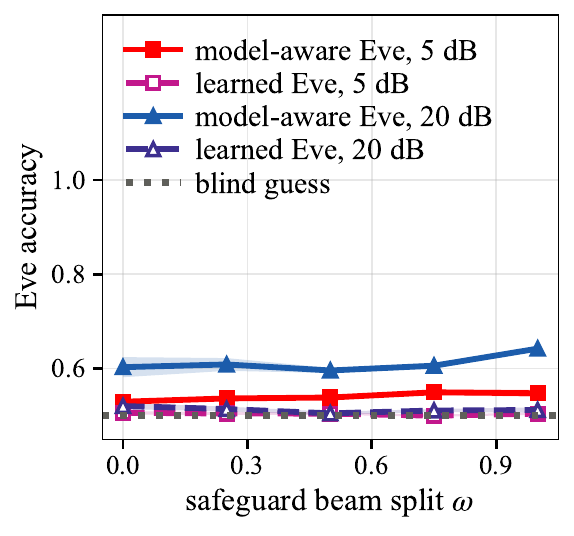}}\hfill
\subfloat[Position PCRB vs. $\omega$ (Eve$_2$).]{\includegraphics[width=0.49\columnwidth]{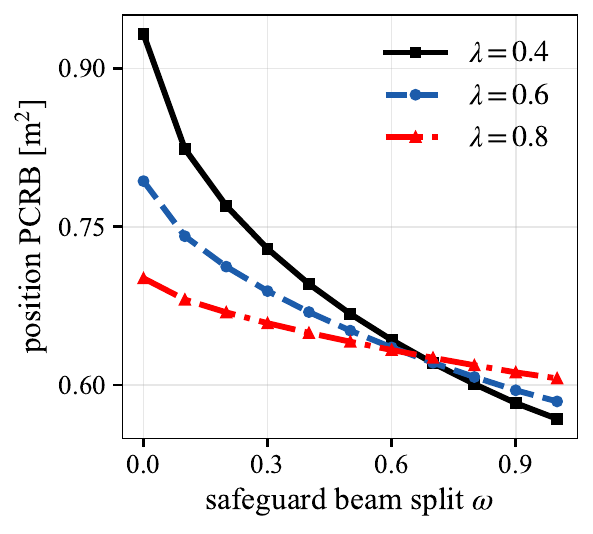}}
\caption{Covertness-sensing trade-off ($N = 8$), traced by the safeguard beam split $\omega$. In (a), the gray cone and its shaded ground footprint mark the directions illuminated by $\tilde\bfv_r$, and the arrows mark $\tilde\bfv_e$ for Eve$_1$ at $(55, -38)$~m and Eve$_2$ at $(-55, -38)$~m. In (b) and (d), every point retrains the generator for $5000$ iterations with Bob's BLER held at $10^{-3}$, and at Eve SNRs of $5$ and $20$~dB. In (c) and (e), the PCRB is shown for $\lambda \in \{0.4, 0.6, 0.8\}$ and depends on Eve's direction through the safeguard beam but not on Eve's SNR.}
\label{fig:tradeoff}
\end{figure}
\subsection{Covertness-Sensing Trade-off}\label{sec:numtradeoff}
Fig.~\ref{fig:tradeoff} traces the covertness-sensing trade-off at $N = 8$ as the safeguard beam split $\omega$ sweeps from full Eve-masking ($\omega = 0$) to full region illumination ($\omega = 1$), retraining the generator at every point with Bob's BLER held at $10^{-3}$. Fig.~\ref{fig:tradeoff}(a) shows the two Eve locations. Eve$_1$ lies inside $\mathcal R$ and inside the main lobe of Bob's beam and models a concealed receiver near the protected asset, whereas Eve$_2$ lies outside both.
For Eve$_1$, both detectors in Fig.~\ref{fig:tradeoff}(b) gain with Eve's SNR at every beam split. The learned Eve stays near $0.53$ at $5$~dB and ranges from $0.66$ to $0.75$ at $20$~dB. The model-aware Eve lies above it, between $0.71$ and $0.73$ at $5$~dB and between $0.84$ and $0.92$ at $20$~dB. The gap between the two shows that the learned Eve under-reports the leakage toward an eavesdropper inside Bob's main lobe. At both SNRs the model-aware accuracy increases with $\omega$, since a larger $\omega$ moves safeguard energy from Eve$_1$ to the region illumination and the masking gain at Eve$_1$ falls from $0.65$ to $0.21$. Full region illumination is therefore not the covertness-optimal choice against Eve$_1$. Since the PCRB in Fig.~\ref{fig:tradeoff}(c) is smallest at $\omega = 1$ and varies between $0.59$ and $0.70$~m$^2$ at $\lambda = 0.6$, covertness favors small $\omega$ while sensing favors large $\omega$, and the beam split selects the operating point on the trade-off.
For the out-of-sector Eve$_2$, the learned Eve in Fig.~\ref{fig:tradeoff}(d) stays within $0.02$ of the blind-guess level at every Eve SNR and every $\omega$, and the model-aware Eve stays between $0.53$ and $0.55$ at $5$~dB and between $0.60$ and $0.64$ at $20$~dB.
 Eve's direction thus determines the leakage. The beam split should therefore favor the region whenever the eavesdropper direction carries weak leakage. Masking Eve$_2$ only costs sensing, and at $\lambda = 0.6$ the PCRB in Fig.~\ref{fig:tradeoff}(e) grows from $0.59$~m$^2$ at $\omega = 1$ to $0.79$~m$^2$ at $\omega = 0$.%

\subsection{Covertness-Reliability Pareto Frontier}\label{sec:numpareto}
Fig.~\ref{fig:pareto} quantifies the reliability cost of covertness for $N = 1$ where the safeguard that masks the transmission at Eve also reaches Bob. We sweep the adversarial weight $\lambda_{\mathrm{adv}}$ of the loss \eqref{eq:loss} over six values in $[0.05, 5]$, each yielding one trained system averaged over ten seeds.
As $\lambda_{\mathrm{adv}}$ grows, the learned Eve in Fig.~\ref{fig:pareto}(a) falls from $0.96$ to $0.53$ and the model-aware Eve from $0.98$ to $0.66$, while Bob's BLER in Fig.~\ref{fig:pareto}(b) rises. The weight $0.05$ keeps Bob at the BLER floor but leaves both Eves above $0.96$.
As can be seen from this figure, the denoiser decoder outperforms the isotropic ML. Also, beyond $\lambda_{\mathrm{adv}} = 2$ the model-aware accuracy improves by less than $0.03$ while the BLER more than doubles, and large weights therefore degrade reliability with negligible gain in covertness.  Fig.~\ref{fig:pareto}(a) also shows the  \emph{retrained Gaussian} which replaces the DDIM by a matched-power Gaussian safeguard at $\lambda_{\mathrm{adv}} = 1$.
Under the retrained Gaussian, Bob's BLER is $1.2\times 10^{-2}$, close to the $1.1\times 10^{-2}$ of the proposed scheme at $\lambda_{\mathrm{adv}} = 5$, yet the model-aware Eve reaches $0.80$ against $0.66$ and the learned Eve $0.74$ against $0.53$. Without a trainable safeguard distribution the covertness objective can shape only the codebook, which yields far less covertness for the same reliability.

\begin{figure}[t]
\centering
\subfloat[Eve accuracy vs. Bob BLER.]{\includegraphics[width=0.49\columnwidth]{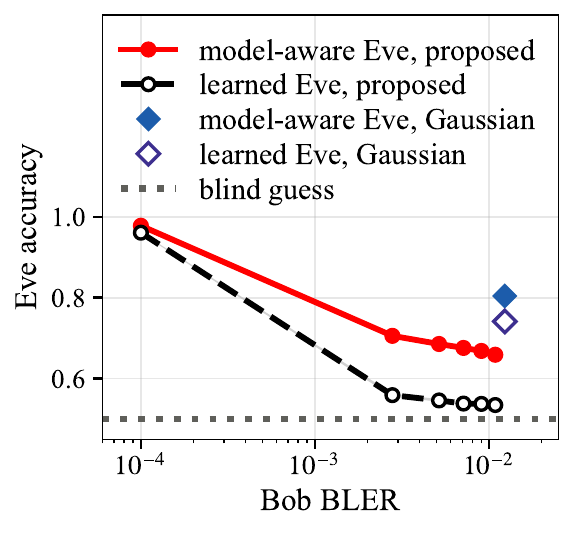}}\hfill
\subfloat[Bob BLER vs. $\lambda_{\mathrm{adv}}$.]{\includegraphics[width=0.49\columnwidth]{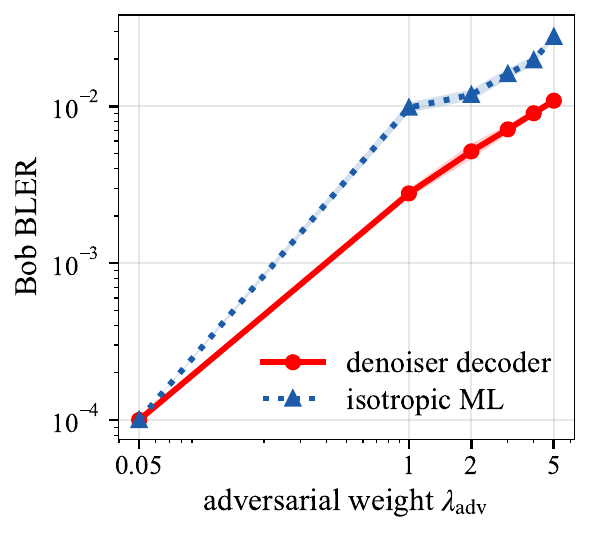}}
\caption{Covertness-reliability trade-off ($N = 1$), by sweeping the adversarial weight $\lambda_{\mathrm{adv}}$ of \eqref{eq:loss}. In (a), Eve's accuracy vs.\ Bob's BLER under the denoiser decoder for $\lambda_{\mathrm{adv}} \in \{0.05, 1, 2, 3, 4, 5\}$ for the proposed scheme and the retrained Gaussian. In (b), Bob's BLER vs.\ $\lambda_{\mathrm{adv}}$ for the denoiser and isotropic ML decoders.}
\label{fig:pareto}
\end{figure}

\subsection{Eavesdropper Channel Quality}\label{sec:numevesnr}
Fig.~\ref{fig:evesnr} varies Eve's SNR from $0$ to $12$~dB at a fixed Bob SNR of $10$~dB and retrains the system at each point. Eve's SNR bin is part of the public tag.
Without a safeguard, the model-aware Eve increases from $0.76$ to $0.87$ and the learned Eve from $0.69$ to $0.80$. With the proposed safeguard the learned Eve stays between $0.54$ and $0.55$, even where Eve's channel is as good as Bob's, whereas the model-aware Eve increases from $0.66$ to $0.74$. The safeguard lowers the model-aware accuracy by $0.11$ to $0.13$ at every SNR, and the Eve's accuracy grows with the SNR because without a shared key  Eve's disadvantage is only her noise level relative to Bob's.  Fig.~\ref{fig:evesnr}(b) shows the reliability cost. As Eve's SNR increases from $0$ to $12$~dB, Bob's BLER under the denoiser decoder increases from $4.4\times 10^{-4}$ to $5.2\times 10^{-3}$, while without a safeguard it stays at the floor of $10^{-4}$.
As Eve's SNR increases, covertness requires the message and dummy codewords to move closer to one another, which makes the decoding harder for Bob.%

\begin{figure}[t]
\centering
\subfloat[Eve accuracy vs. Eve SNR.]{\includegraphics[width=0.49\columnwidth]{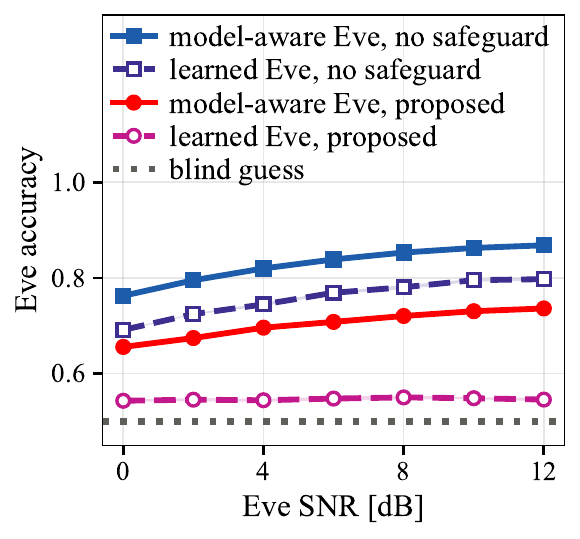}}\hfill
\subfloat[Bob BLER vs. Eve SNR.]{\includegraphics[width=0.49\columnwidth]{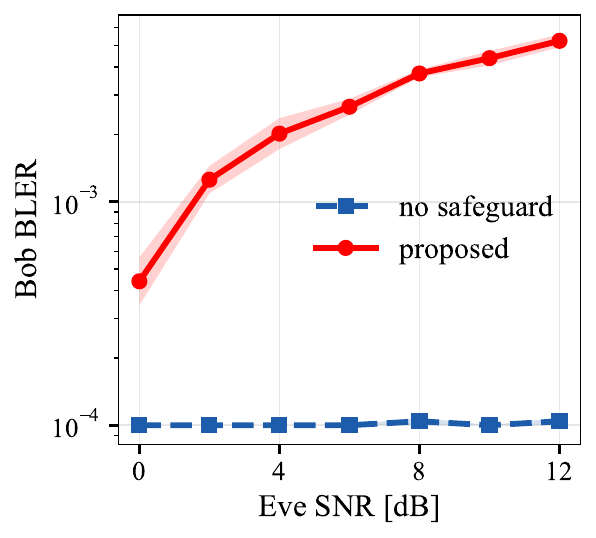}}
\caption{Effect of Eve's SNR at a Bob SNR of $10$~dB ($N = 1$). In (a), Eve's detection accuracy is shown for the proposed safeguard and no safeguard, averaged over ten seeds. In (b), Bob's BLER is shown for the denoiser decoder, and for the isotropic ML without a safeguard. }%
\label{fig:evesnr}
\end{figure}

\subsection{Receiver Comparison}\label{sec:numreceivers}
Fig.~\ref{fig:receivers}(a) compares the seed-free receivers of Section~\ref{sec:receivers} at $N = 1$ for Bob SNRs from $4$ to $14$~dB after training at $10$~dB.
The denoiser decoder changes from $0.17$ at $4$~dB to below $10^{-4}$ at $14$~dB. The isotropic ML, which ignores the structure of the safeguard, reaches $1.1\times 10^{-2}$ at $10$~dB and $1.1\times 10^{-3}$ at $14$~dB.
 Fig.~\ref{fig:receivers}(b) varies the codebook size $M_a$ at the operating point $\lambda = 0.6$ and $\gamma = 0.5$ of Fig.~\ref{fig:pareto}. The two rules coincide up to $M_a = 8$. Beyond, the isotropic ML stays between $3.0\times 10^{-3}$ and $5.9\times 10^{-3}$, whereas the denoiser decoder falls to  $7.9\times 10^{-4}$ at $M_a = 32$ and $3.3\times 10^{-4}$ at $M_a = 64$. At this operating point the safeguard reaches Bob with the same power as the codewords. A larger codebook moves the codewords closer. The isotropic ML treats the safeguard as noise and decodes with a higher BLER, whereas the denoiser decoder reconstructs the safeguard, subtracts the reconstruction, and decodes the codeword against the thermal noise only.

\begin{figure}[t]
\centering
\subfloat[$N = 1$.]{\includegraphics[width=0.49\columnwidth]{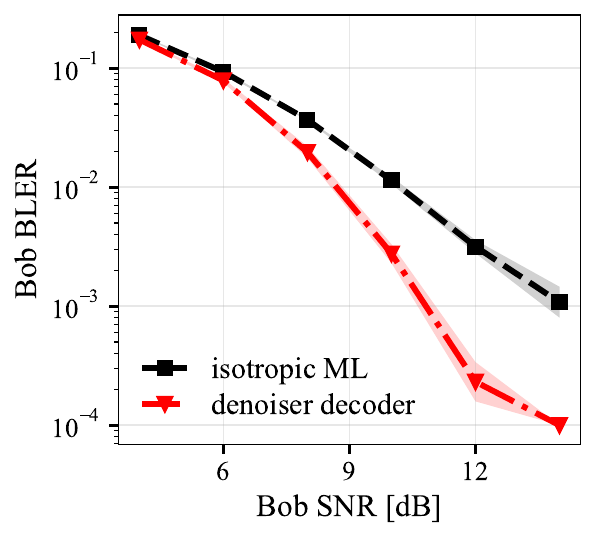}}\hfill
\subfloat[Codebook size, $N = 1$.]{\includegraphics[width=0.49\columnwidth]{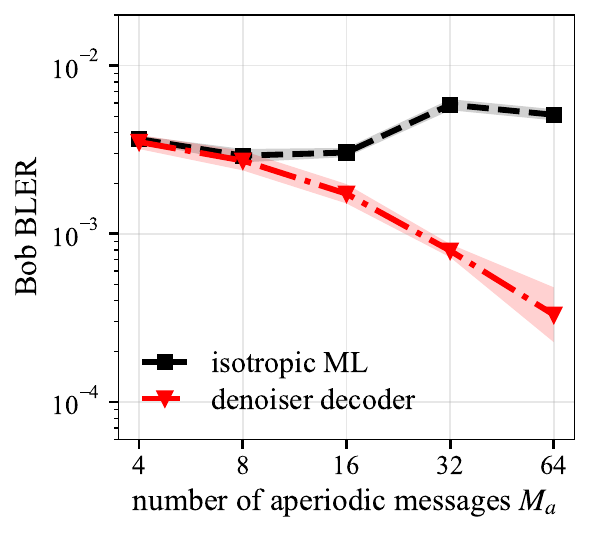}}
\caption{Bob's aperiodic BLER for the seed-free receivers (trained at $10$~dB). (a) The two rules at $N = 1$ versus Bob's SNR; (b) the two rules at $N = 1$ versus the codebook size $M_a$ at $\lambda = 0.6$ and $\gamma = 0.5$.}%
\label{fig:receivers}
\end{figure}

\subsection{Aperiodic Rate}\label{sec:numrate}
Fig.~\ref{fig:rate} varies the aperiodic rate of problem \eqref{eq:opt} through the codebook size $M_a \in \{4, 8, 16, 32, 64\}$ with an equal-size dummy set and $n = 64$. For each $M_a$, the system is retrained at every operating point of $\lambda \in \{0.4, 0.6, 0.8\}$, $\gamma \in \{0.3, 0.5\}$, $\omega \in \{0.5, 0.7\}$ that meets the sensing budget $\epsilon_s = 0.70$~m$^2$, and the figure shows the design with the lowest model-aware accuracy among those with a BLER of at most $\bar\epsilon_a = 3\times 10^{-3}$. In Fig.~\ref{fig:rate}(a), the learned Eve stays between $0.50$ and $0.55$, while the model-aware Eve falls from $0.67$ and $0.68$ at $M_a = 4$ for $N = 1$ and $N = 8$ to between $0.63$ and $0.66$ for $M_a$ from $8$ to $32$ and rises to $0.70$ and $0.71$ at $M_a = 64$. In Fig.~\ref{fig:rate}(b), Bob's BLER decreases from $2.5\times 10^{-3}$ and $3.0\times 10^{-3}$ at $M_a = 4$ for $N = 1$ and $N = 8$ to below $10^{-4}$ at $M_a = 64$, because the cross-entropy term of the loss \eqref{eq:loss} grows with the number of messages.

\begin{figure}[t]
\centering
\subfloat[Eve accuracy vs. $M_a$.]{\includegraphics[width=0.49\columnwidth]{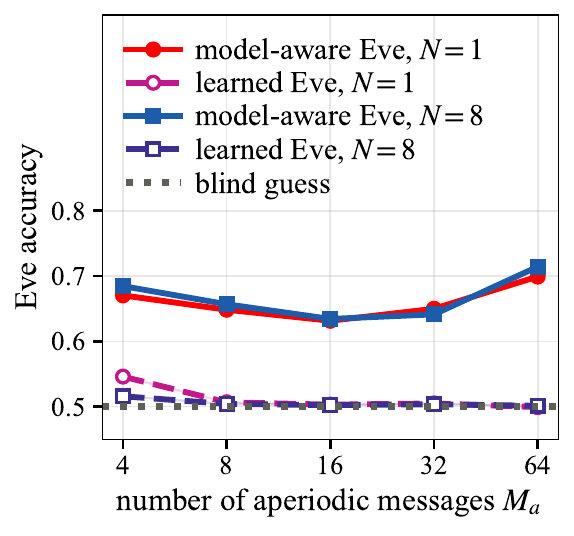}}\hfill
\subfloat[Bob BLER vs. $M_a$.]{\includegraphics[width=0.49\columnwidth]{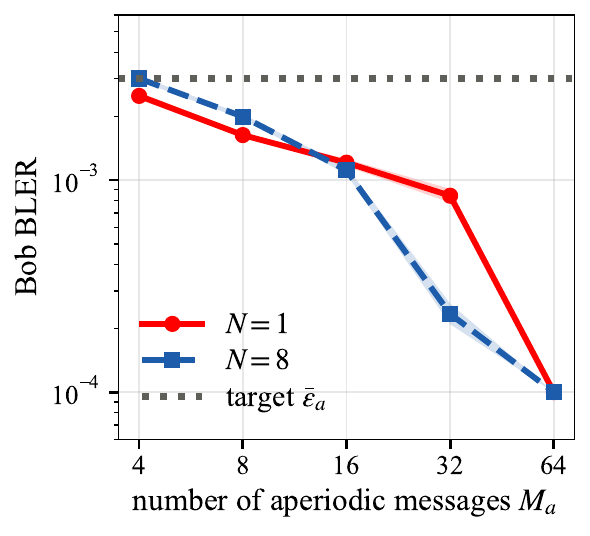}}
\caption{Effect of the number of aperiodic messages $M_a$ on (a) covertness and (b) reliability for $N = 1$ and $N = 8$.}%
\label{fig:rate}
\end{figure}

\subsection{Robustness to Eve Location Error}\label{sec:numrobust}
The transmit design relies on Eve's assumed direction through the safeguard beam \eqref{eq:beams}, and Fig.~\ref{fig:evecsi} quantifies the sensitivity with $N = 8$ at $\omega = 0.5$. We use the uncertainty model \eqref{eq:evecsimodel} with $\Delta_e = 0.1$ and evaluate the design pointwise in the true error $\delta_e \in [-0.1, 0.1]$, where $\delta_e = -0.1$ moves Eve close to Bob's direction. The \emph{nominal design} is trained with Eve at the nominal direction $\hat u_e$.
Under this design, the model-aware Eve lies at $0.69$ at $\delta_e = 0$, $0.81$ at $\delta_e = -0.1$, and $0.53$ at $\delta_e = 0.1$. As shown in Fig.~\ref{fig:evecsi}(b), when Eve's true direction approaches Bob's, the information amplitude $|\alpha_c|$ at Eve increases and the safeguard amplitude $|\alpha_s|$ decreases because of the null at Bob, and Eve receives a stronger signal with weaker masking.
\emph{Robust training} draws $\delta_e$ uniformly from $[-\Delta_e, \Delta_e]$ at every training slot and therefore enforces covertness on average over the errors, whereas the evaluation is pointwise in $\delta_e$. The two designs coincide within $0.01$ at every offset. When Eve is close to Bob's direction, no beamforming can separate the two, and the protection must come from the safeguard, as in the single-antenna case.

\begin{figure}[t]
\centering
\subfloat[Eve accuracy vs. $\delta_e$.]{\includegraphics[width=0.49\columnwidth]{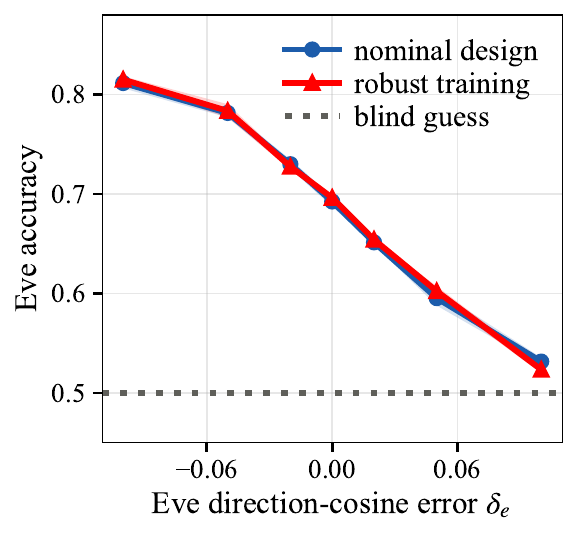}}\hfill
\subfloat[Leakage amplitudes vs. $\delta_e$.]{\includegraphics[width=0.49\columnwidth]{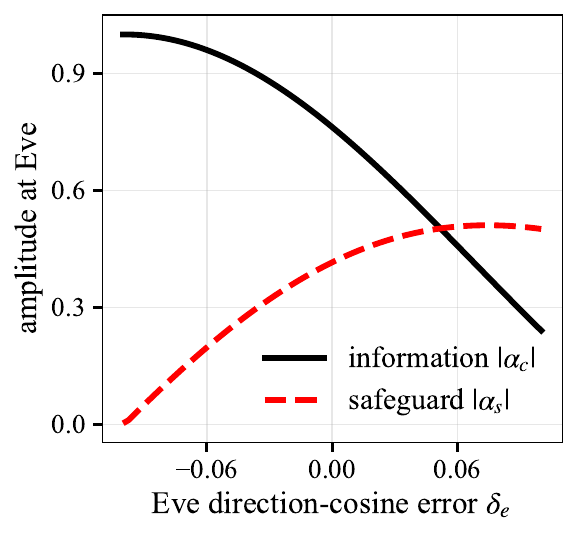}}
\caption{Sensitivity to an error $\delta_e$ in Eve's assumed direction cosine. In (a), Eve's accuracy for the \emph{nominal design} and \emph{robust training}, with the model-aware Eve, averaged over three seeds. In (b), the information and safeguard amplitudes $|\alpha_c|$ and $|\alpha_s|$ at Eve's true direction. }%
\label{fig:evecsi}
\end{figure}

\section{Conclusion}\label{sec:conclusion}
We have developed a keyless diffusion-based covert transmission framework for a MISO UAV-ISAC system with real-time target tracking. The framework integrates three design layers, the sensing, security, and beamforming layers. The sensing layer tracks the target with an EKF and triggers the aperiodic traffic by an estimator-independent rule. The security layer superposes a message-agnostic DDIM safeguard whose trained distribution is public and makes triggered and non-triggered slots statistically similar at the eavesdropper. The beamforming layer splits the safeguard between region illumination and eavesdropper masking. We have formulated the joint design of the three layers as a rate maximization under covertness, reliability, and sensing constraints and solved the problem by end-to-end training. The training enforces the covertness constraint through the MMD and an adversarial classifier, and the legitimate user decodes the aperiodic messages without a shared key. In the numerical results, the trigger reproduced the target's presence in $120$ of $123$ in-region slots.
With $8$ antennas and an Eve SNR of $5$~dB, the learned Eve stayed near the blind-guess level of $0.5$ for every safeguard beam split, a model-aware Eve reached $0.71$ to $0.73$ for an eavesdropper inside Bob's main lobe and $0.53$ to $0.55$ outside it, and the average position PCRB varied only between $0.59$ and $0.70$~m$^2$. For future work, detectors that combine observations across many slots \cite{Bash2013} and scenarios with multiple targets \cite{Ruan2025} and users are of interest. A further direction for future study is an information-theoretic analysis of generative randomized encoders that connects the trained safeguard distribution to bounds on the achievable covert rate \cite{Bloch2016}.


\begin{thebibliography}{30}

\bibitem{Liu2022}
F.~Liu et al., ``Integrated sensing and communications: Toward dual-functional wireless networks for 6G and beyond,'' \emph{IEEE J. Sel. Areas Commun.}, vol.~40, no.~6, pp.~1728--1767, June~2022.

\bibitem{Visa2024}
V. Koivunen \emph{et al.}, ``Multicarrier ISAC: Advances in waveform design, signal processing, and learning under nonidealities,'' \emph{IEEE Signal Process. Mag.}, vol.~41, no.~5, pp.~17--30, Sept. 2024.

\bibitem{Ahmadi2026}
M. J. Ahmadi, R. F. Schaefer, and H. V. Poor, ``Integrated sensing and communications for unsourced random access: Fundamental limits and practical model,'' \emph{IEEE Trans. Wireless Commun.}, vol.~25, pp.~16547--16561, 2026.

\bibitem{NikbakhtISIT24}
H. Nikbakht, M. Wigger, S. Shamai (Shitz), and H. V. Poor, ``Integrated sensing and communication in the finite blocklength regime,'' in \emph{Proc. IEEE Int. Symp. Inf. Theory (ISIT)}, Athens, Greece, 2024, pp. 2790--2795.

\bibitem{Qin2025}
P. Qin \emph{et al.}, ``URLLC-aware trajectory plan and beamforming design for NOMA-aided UAV integrated sensing, communication, and computation networks,'' \emph{IEEE Trans. Veh. Technol.}, vol.~74, no.~1, pp.~1610--1625, Jan. 2025.

\bibitem{Luat2025}
L. B. Luat, N. C. Luong and D. I. Kim, ``Integrated radar and communication in ultra-reliable and low-latency communications-enabled UAV networks,'' \emph{IEEE Trans. Veh. Technol.}, vol.~74, no.~8, pp.~13133--13138, Aug. 2025.

\bibitem{NikbakhtISIT25}
H. Nikbakht, Y. C. Eldar, and H. V. Poor, ``A MIMO ISAC system for ultra-reliable and low-latency communications,'' in \emph{Proc. IEEE Int. Symp. Inf. Theory (ISIT)}, Ann Arbor, MI, USA, 2025, pp. 1--6.

\bibitem{Ahmed2025}
M. Ahmed \emph{et al.}, ``Advancements in UAV-based integrated sensing and communication: A comprehensive survey,'' \emph{IEEE Open J. Commun. Soc.}, vol.~7, pp.~9145--9182, 2026.

\bibitem{Liu2020PB}
F. Liu, W. Yuan, C. Masouros, and J. Yuan, ``Radar-assisted predictive beamforming for vehicular links: Communication served by sensing,'' \emph{IEEE Trans. Wireless Commun.}, vol.~19, no.~11, pp.~7704--7719, Nov. 2020.

\bibitem{NikbakhtAsilomar24}
H. Nikbakht, M. Wigger, S. Shamai (Shitz), and H. V. Poor, ``A memory-based reinforcement learning approach to integrated sensing and communication,'' in \emph{Proc. Asilomar Conf. Signals, Syst., Comput.}, Pacific Grove, CA, USA, 2024, pp. 433--437.

\bibitem{Su2021}
N. Su, F. Liu, and C. Masouros, ``Secure radar-communication systems with malicious targets: Integrating radar, communications and jamming functionalities,'' \emph{IEEE Trans. Wireless Commun.}, vol.~20, no.~1, pp.~83--95, Jan. 2021.

\bibitem{Boroujeni2026}
A. Khandan Boroujeni \emph{et al.}, ``Frequency hopping waveform design for secure integrated sensing and communications,'' \emph{IEEE Trans. Inf. Forensics Security}, vol.~21, pp.~3719--3734, 2026.

\bibitem{Goel2008}
S. Goel and R. Negi, ``Guaranteeing secrecy using artificial noise,'' \emph{IEEE Trans. Wireless Commun.}, vol.~7, no.~6, pp.~2180--2189, June 2008.

\bibitem{Bash2013}
B. A. Bash, D. Goeckel, and D. Towsley, ``Limits of reliable communication with low probability of detection on AWGN channels,'' \emph{IEEE J. Sel. Areas Commun.}, vol.~31, no.~9, pp.~1921--1930, Sept. 2013.

\bibitem{Zhou2019}
X. Zhou \emph{et al.}, ``Joint optimization of a UAV's trajectory and transmit power for covert communications,'' \emph{IEEE Trans. Signal Process.}, vol.~67, no.~16, pp.~4276--4290, Aug. 2019.

\bibitem{Tang2024}
R. Tang \emph{et al.}, ``Dual-functional artificial noise (DFAN) aided robust covert communications in integrated sensing and communications,'' \emph{IEEE Trans. Commun.}, vol.~73, no.~2, pp.~1072--1086, Feb. 2025.

\bibitem{Yao2025}
J. Yao \emph{et al.}, ``UAV-enabled secure ISAC against dual eavesdropping threats: Joint beamforming and trajectory design,'' \emph{IEEE Wireless Commun. Lett.}, vol.~14, no.~10, pp.~3199--3203, Oct. 2025.

\bibitem{Song2021}
J. Song, C. Meng, and S. Ermon, ``Denoising diffusion implicit models,'' in \emph{Proc. Int. Conf. Learn. Represent.}, 2021.

\bibitem{Ho2020}
J. Ho, A. Jain, and P. Abbeel, ``Denoising diffusion probabilistic models,'' in \emph{Adv. Neural Inf. Process. Syst.}, vol.~33, pp.~6840--6851, 2020.

\bibitem{ZhaoGenAI2024}
C. Zhao \emph{et al.}, ``Generative AI for secure physical layer communications: A survey,'' \emph{IEEE Trans. Cogn. Commun. Netw.}, vol.~11, no.~1, pp.~3--26, Feb. 2025.

\bibitem{Wang2025}
J. Wang \emph{et al.}, ``Generative AI based secure wireless sensing for ISAC networks,'' \emph{IEEE Trans. Inf. Forensics Security}, vol.~20, pp.~5195--5210, 2025.

\bibitem{Gretton2012}
A. Gretton \emph{et al.}, ``A kernel two-sample test,'' \emph{J. Mach. Learn. Res.}, vol.~13, pp.~723--773, 2012.

\bibitem{Goodfellow2014}
I. Goodfellow \emph{et al.}, ``Generative adversarial nets,'' in \emph{Adv. Neural Inf. Process. Syst.}, vol.~27, pp.~2672--2680, 2014.

\bibitem{Stove2004}
A. G. Stove, A. L. Hume, and C. J. Baker, ``Low probability of intercept radar strategies,'' \emph{IEE Proc. Radar Sonar Navig.}, vol. 151, no. 5, pp. 249--260, Oct. 2004.

\bibitem{Pace2009}
P. E. Pace, \emph{Detecting and Classifying Low Probability of Intercept Radar}, 2nd ed. Norwood, MA, USA: Artech House, 2008.

\bibitem{Stoica2007}
P. Stoica, J. Li, and Y. Xie, ``On probing signal design for MIMO radar,'' \emph{IEEE Trans. Signal Process.}, vol.~55, no.~8, pp.~4151--4161, Aug. 2007.

\bibitem{Richards2014}
M. A. Richards, \emph{Fundamentals of Radar Signal Processing}, 3rd ed., New York, NY, USA: McGraw-Hill, 2022.

\bibitem{Kay1993}
S. M. Kay, \emph{Fundamentals of Statistical Signal Processing: Estimation Theory}. Upper Saddle River, NJ, USA: Prentice-Hall, 1993.

\bibitem{Poor}
H. V. Poor, \emph{An Introduction to Signal Detection and Estimation}, Springer Science \& Business Media, 2013.

\bibitem{BarShalom2001}
Y. Bar-Shalom, X. R. Li, and T. Kirubarajan, \emph{Estimation with Applications to Tracking and Navigation}, New York, NY, USA: Wiley, 2001.

\bibitem{Tichavsky1998}
P. Tichavsk\'y, C. H. Muravchik, and A. Nehorai, ``Posterior Cram\'er-Rao bounds for discrete-time nonlinear filtering,'' \emph{IEEE Trans. Signal Process.}, vol.~46, no.~5, pp.~1386--1396, May 1998.

\bibitem{Jiang2024}
Y. Jiang, Q. Wu, W. Chen, and K. Meng, ``UAV-enabled integrated sensing and communication: Tracking design and optimization,'' \emph{IEEE Commun. Lett.}, vol.~28, no.~5, pp.~1024--1028, May 2024.

\bibitem{DiDonato1961}
A. R. DiDonato and M. P. Jarnagin, ``Integration of the general bivariate Gaussian distribution over an offset circle,'' \emph{Math. Comput.}, vol.~15, no.~76, pp.~375--382, 1961.

\bibitem{Polyanskiy2010}
Y. Polyanskiy, H. V. Poor, and S. Verd\'u, ``Channel coding rate in the finite blocklength regime,'' \emph{IEEE Trans. Inf. Theory}, vol.~56, no.~5, pp.~2307--2359, May 2010.

\bibitem{Nikbakht2025}
H. Nikbakht, Y. C. Eldar, and H. V. Poor, ``An integrated sensing and communication system for time-sensitive targets with random arrivals,'' \emph{IEEE J. Sel. Areas Commun.}, vol. 44, pp. 3332--3348, 2026.

\bibitem{Chung2023}
H. Chung \emph{et al.}, ``Diffusion posterior sampling for general noisy inverse problems,'' in \emph{Proc. Int. Conf. Learn. Represent.}, 2023.

\bibitem{Vincent2011}
P. Vincent, ``A connection between score matching and denoising autoencoders,'' \emph{Neural Comput.}, vol.~23, no.~7, pp.~1661--1674, July 2011.

\bibitem{Efron2011}
B. Efron, ``Tweedie's formula and selection bias,'' \emph{J. Amer. Statist. Assoc.}, vol.~106, no.~496, pp.~1602--1614, 2011.

\bibitem{OShea2017}
T. O'Shea and J. Hoydis, ``An introduction to deep learning for the physical layer,'' \emph{IEEE Trans. Cogn. Commun. Netw.}, vol.~3, no.~4, pp.~563--575, Dec. 2017.

\bibitem{Boyd2004}
S. Boyd and L. Vandenberghe, \emph{Convex Optimization}. Cambridge, UK: Cambridge University Press, 2004.

\bibitem{Kozachenko1987}
L. F. Kozachenko and N. N. Leonenko, ``Sample estimate of the entropy of a random vector,'' \emph{Probl. Inf. Transm.}, vol.~23, no.~2, pp.~95--101, 1987.

\bibitem{HoSalimans2022}
J. Ho and T. Salimans, ``Classifier-free diffusion guidance,'' in \emph{Proc. NeurIPS Workshop Deep Gener. Models Downstream Appl.}, 2021.

\bibitem{Yang2022}
H. Yang, E. Liang, M. Pan, and R. D. Wesel, ``CRC-aided list decoding of convolutional codes in the short blocklength regime,'' \emph{IEEE Trans. Inf. Theory}, vol.~68, no.~6, pp.~3744--3766, June 2022.

\bibitem{Perez2018}
E. Perez \emph{et al.}, ``FiLM: Visual reasoning with a general conditioning layer,'' in \emph{Proc. AAAI Conf. Artif. Intell.}, 2018.

\bibitem{Ruan2025}
H. Ruan, H. Nikbakht, R. Zhang, H. Chen, and Y. C. Eldar, ``Near-field integrated sensing and communication for multi-target indication,'' in \emph{Proc. IEEE 36th Int. Symp. Pers., Indoor Mobile Radio Commun. (PIMRC)}, Istanbul, Turkiye, 2025, pp.~1--6.

\bibitem{Bloch2016}
M. R. Bloch, ``Covert communication over noisy channels: A resolvability perspective,'' \emph{IEEE Trans. Inf. Theory}, vol.~62, no.~5, pp.~2334--2354, May 2016.

\end{thebibliography}
\end{document}